\documentclass[pdflatex,sn-mathphys-num]{sn-jnl}

\usepackage{graphicx}
\usepackage{multirow}
\usepackage{amsmath,amssymb,amsfonts}
\usepackage{amsthm}
\usepackage{mathrsfs}
\usepackage[title]{appendix}
\usepackage{xcolor}
\usepackage{textcomp}
\usepackage{manyfoot}
\usepackage{booktabs}
\usepackage{algorithm}
\usepackage{algorithmicx}
\usepackage{algpseudocode}
\usepackage{listings}
\usepackage{pgfplots}
\pgfplotsset{compat=1.18}

\theoremstyle{thmstyleone}
\newtheorem{theorem}{Theorem}
\newtheorem{proposition}[theorem]{Proposition}

\theoremstyle{thmstyletwo}

\newtheorem{remark}{Remark}

\theoremstyle{thmstylethree}

\newcommand{\ket}[1]{|#1\rangle}
\newcommand{\bra}[1]{\langle#1|}
\newcommand{\braket}[2]{\langle#1|#2\rangle}

\newcommand{\Feff}{F_{\text{eff}}}
\newcommand{\epsleaf}{\epsilon_{\mathrm{leaf}}}
\newcommand{\epstot}{\epsilon_{\mathrm{tot}}}

\begin{document}


\title[Adaptive Multi-Backend Near-Clifford Simulation]{Adaptive Multi-Backend Simulation of Near-Clifford Quantum Circuits via Spatial Stabilizer-Frame Partitioning}

\author*[1]{\fnm{H\'{e}ctor J.} \sur{Garc\'{i}a}}\email{hjgarcia@umich.edu}

\affil*[1]{\orgdiv{Electrical Engineering and Computer Science},
  \orgname{University of Michigan},
  \orgaddress{\city{Ann Arbor}, \state{MI} \postcode{48109}, \country{USA}}}

\abstract{We present an exact amplitude simulator for Clifford$+$T
quantum circuits that combines a Feynman path sum across a balanced
qubit bipartition with stabilizer-frame simulation on each half. The
construction extends prior stabilizer-based Schr\"odinger-Feynman
methods in three directions: recursive multilevel bipartition into a
binary tree, automatic fallback to dense state-vector simulation when
a leaf's stabilizer frame would exceed its memory ceiling, and a
cost-model-driven partition selector that replaces the standard
cut-count minimization heuristic.

We show cut-count minimization is an unreliable proxy in practice: a
globally cleaner partition can reduce cross-cut count yet
\emph{increase} wall-clock time, because it imbalances T-gate density
across halves and inflates per-half stabilizer-frame size. Our cost
model substitutes the stabilizer-frame bound $2^{w}$ for the dense
$2^{n}$ ceiling per side and explicitly models per-amplitude readout
cost; isolating that term uncovered a quadratic-asymptotic
inefficiency in the leaf simulator's end-of-path amplitude
extraction, fixed by replacing it with an existing
$O(F \cdot s \cdot n)$ single-amplitude inner product.

On a structured hierarchical $n{=}16$ benchmark the recursive
simulator beats monolithic stabilizer-frame simulation by $92\times$
to $17{,}645\times$, wins by $79\times$ per path against a dense
half-state-vector baseline under an identical cut, and beats a
production state-vector simulator end to end by up to $47.9\times$
(median ${\sim}5\times$). On adversarial random Clifford$+$T circuits
the dense state vector wins past a crossover near $n/2$ cross-cut
gates -- the regime the cost model identifies. The dominant cost, the
cross-cut Feynman sum, is embarrassingly parallel with constant
inter-worker communication, unlike recent matrix-product-state
stabilizer-tensor methods whose inner contraction loop is
sequential.}


\maketitle


\section{Introduction}\label{sec:intro}

Stabilizer-frame
simulation~\cite{garcia2012quipu,garcia2015simulation,garcia2014inner}
extends the Gottesman--Knill formalism~\cite{gottesman1998heisenberg,
aaronson2004improved} to circuits with a small number of non-Clifford
gates by representing the state as a sum of phase-corrected stabilizer
states sharing a single tableau.  Clifford gates update the tableau
alone at cost $O(n^2)$; each T-gate \emph{cofactors} the frame,
potentially doubling its population.  When cofactors collapse, the
effective frame dimension stays well below the worst case $2^t$; on
many structured workloads the technique is both faster and more memory
efficient than state-vector simulation~\cite{boixo2018characterizing}
and stabilizer-rank
decomposition~\cite{bravyi2016improved,bravyi2018simulation}.

The wall comes when T-gates are too numerous for the merger to keep
up.  Once $\Feff \to 2^t$ the per-gate work dominates the qubit count
$n$, and the wall arrives at $t$ small relative to $n$.  No
representational trick within the frame itself can move it: the
exponential dependence is on the magic, not the entanglement.

To push beyond, we observe that T-gates in many near-Clifford
workloads are not distributed uniformly: fault-tolerant
gadgets~\cite{bravyi2018simulation}, hierarchical
algorithms~\cite{coppersmith1994approximate}, and most magic-state
sub-routines concentrate T-gates within identifiable sub-registers.
A spatial bipartition of the qubit set into halves $A$ and $B$ thus
splits the global T-count into two smaller per-half counts $T_A$ and
$T_B$, with each half representable by its own much smaller stabilizer
frame.  The cross-cut entangling Clifford gates are then summed as a
Feynman path
integral~\cite{markov2008simulating,pednault2017pareto}---a
partitioning technique introduced by Markov and Shi for tensor
networks and deployed at scale by qsim~\cite{qsim2020} for state-vector
leaves in the Sycamore verification.  The same partitioning composes
cleanly with stabilizer-frame leaves; the closest prior work in this
direction~\cite{huang2021feynman} sketches a related composition with
a stabilizer-projector contraction leaf but offers neither an
implementation nor empirical results.

Choosing the bipartition itself turns out to be a non-trivial design
decision.  The classical proxy---minimize the number of two-qubit
gates that cross the cut, e.g.\ via Kernighan--Lin~\cite{kl1970} or
spectral graph-partition~\cite{pothen1990partitioning}---is correct in
limit but fails in practice: a globally cleaner partition can shift
all T-gates onto one half and inflate that half's stabilizer-frame
population by far more than the cross-cut-count reduction saves.  An
algorithm that picks partitions without modelling per-half frame size
gives the stabilizer leaves configurations they cannot
compress, defeating the whole purpose of the cut.

\paragraph{Contribution.}  We extend stabilizer-frame simulation along
this spatial axis with both an algorithmic construction and a
cost-model that drives its automated tuning.  Specifically:

\begin{enumerate}
  \item We define an algorithm that bipartitions an $n$-qubit
        near-Clifford circuit, simulates each half with an independent
        stabilizer frame, and sums their per-amplitude products over
        the $2^d$ Feynman paths through the $d$ cross-cut Clifford
        gates.
  \item We extend the construction to a \emph{multilevel binary tree
        of bipartitions}, multiplying the leaf-cost compression at
        each level.
  \item We make leaf representation \emph{adaptive}: each leaf is
        evaluated with whichever of stabilizer-frame or dense
        state-vector simulation is predicted cheapest from its own
        T-count and qubit count.  This worst-case guarantee---the
        algorithm is never asymptotically slower than direct dense
        simulation, on any input---is what makes the cut safe to
        apply without prior knowledge of T-gate placement, and
        generalizes naturally to additional leaf representations.
  \item We make \emph{partition selection} cost-model-driven.  An
        operation-count predictor substitutes the stabilizer-frame
        bound $\Feff \le 2^w$ for the dense $2^n$ ceiling per side,
        tracks frame growth across T-gates that fire after the two
        halves are first merged, and includes an explicit end-of-path
        dense-state-vector materialization term that accounts for the
        slice-vs-inline asymmetry.  Two micro-benchmarked constants---per-event
        tensor-and-apply cost and per-op readout cost relative to gate
        apply---are measured directly on the target hardware, so no
        workload-tuned multiplier is required.  Candidate partitions
        from Kernighan--Lin, spectral bisection, and multi-start
        variants are scored jointly with per-cross-cut
        merge-or-decompose decisions.
  \item We analyze the parallel structure: the algorithm has three
        nested embarrassingly parallel axes (Feynman paths, per-path
        leaves, intra-frame work) with $O(1)$ inter-worker
        communication.
  \item We report empirical results on two workload classes.  On a
        structured hierarchical $n{=}16$ benchmark ($81$ measurements
        over $27$ cross-cut-budget cells) the recursive simulator beats
        monolithic stabilizer-frame simulation by $92\times$ to
        $17{,}645\times$ and qsim~\cite{qsim2020} by $1.3\times$ to
        $47.9\times$ (median ${\sim}5\times$); isolating the per-path
        leaf cost against qsimh gives a $79\times$ per-path advantage.
        On adversarial random Clifford$+$T circuits the dense state
        vector wins past a crossover near $n/2$ cross-cut gates, the
        regime the cost model identifies.
\end{enumerate}

\paragraph{Outline.}  Section~\ref{sec:background} reviews the
stabilizer-frame simulator and the partitioning structure we adopt
from state-vector simulation.  Section~\ref{sec:algorithm} specifies
the algorithm, its multilevel form, the adaptive leaf rule, and the
cost-model-driven partition selector.  Section~\ref{sec:cost} states
four propositions characterizing its cost---including a formal
statement of the $\Feff$-aware predictor and the two-tier calibration
that turns it into a wall-clock predictor.  Section~\ref{sec:parallel}
analyzes the parallel structure.  Section~\ref{sec:eval} reports
empirical results, including the discrepancy between cross-cut count
and wall-clock time that motivates the cost model.
Section~\ref{sec:discussion} discusses limitations, related
approaches, and future work.


\section{Background}\label{sec:background}

This section reviews the two ingredients Quipu-Cut combines: the
stabilizer-frame simulator that evaluates each leaf
(Section~\ref{sec:stabframe}), and the hybrid Schr\"odinger--Feynman
partitioning, adopted from state-vector simulation, that joins the
leaves over a Feynman path sum.

\subsection{Stabilizer-frame simulation}\label{sec:stabframe}

A stabilizer state on $n$ qubits is a joint $+1$ eigenstate of $n$
commuting Pauli operators, the \emph{stabilizer
group}~\cite{gottesman1998heisenberg}.  It is determined by an
$n \times (2n)$ binary tableau and admits an $O(n^2)$ Clifford-gate
update rule~\cite{aaronson2004improved}.  Stabilizer states span the
Clifford-reachable orbit of the computational basis but cannot
represent arbitrary superpositions.

A \emph{stabilizer frame}~\cite{garcia2015simulation} encodes an
arbitrary quantum state $\ket{\psi}$ as
\begin{equation}
  \ket{\psi} \;=\; \sum_{i=1}^{F} c_i \, P_i \ket{\psi_S}
  \label{eq:stabframe}
\end{equation}
where $\ket{\psi_S}$ is a single shared stabilizer state encoded by an
$n \times 2n$ tableau, $\{P_i\}$ is a set of Pauli ``correction''
operators encoded as phase-vector/ amplitude pairs, and the $c_i$ are
complex coefficients.  Clifford gates act on $\ket{\psi_S}$ alone via
the Aaronson--Gottesman update, costing $O(n^2)$ regardless of $F$.
T~gates fail to be Clifford and \emph{cofactor} the frame: each frame
element splits into two new elements (one per eigenvalue branch),
potentially doubling $F$.  A merge pass collapses elements whose
correction operators coincide up to phase, and the surviving
\emph{effective} frame dimension $\Feff$ can be far smaller than the
worst-case $2^t$.

For purposes of this paper, the salient facts about
stabilizer frames are three:
\begin{itemize}
  \item Single-amplitude evaluation $\braket{x}{\psi}$ costs
        $O(F \cdot n^2)$, dominated by the $F$ stabilizer-state
        inner products~\cite{garcia2014inner}.
  \item Frame growth from $t$ T-gates is bounded by $2^t$, with
        $\Feff \ll 2^t$ on structured circuits but $\Feff = 2^t$
        in adversarial cases (random T-gate placement).
  \item Two independent stabilizer frames on disjoint qubit subsets
        compose via tensor product without coupling their tableaus,
        making the frame interface a natural leaf representation for
        a bipartition-based hybrid.
\end{itemize}

\subsection{Hybrid Schr\"odinger--Feynman simulation}
\label{sec:hsf}

Markov and Shi~\cite{markov2008simulating} introduced a hybrid
in which the $n$ qubits are bipartitioned into sets $A$ and $B$
($|A|+|B| = n$).  Clifford and non-Clifford gates that act only on $A$
or only on $B$ are applied to the respective half-simulators.  Gates
that span the cut---we call these \emph{cross-cut gates}---cannot be
applied directly.  Following Markov--Shi, a two-qubit gate $U_{AB}$
spanning the cut is decomposed into a sum of separable terms,
\begin{equation}
  U_{AB} \;=\; \sum_{j=1}^{r} U_A^{(j)} \otimes U_B^{(j)},
  \label{eq:gate-decomp}
\end{equation}
where $r$ is the Schmidt rank of $U_{AB}$.  For the controlled-Z gate
$r=2$.  Choosing one term per cross-cut gate yields a \emph{Feynman
path} $\alpha = (\alpha_1, \ldots, \alpha_d) \in \{1,\ldots,r\}^d$
through the $d$ cross-cut gates; the global amplitude is the sum
\begin{equation}
  \braket{x}{U \ket{0^n}} \;=\; \sum_{\alpha} \, \mathrm{phase}(\alpha)
    \cdot \braket{x_A}{U_A^{(\alpha)} \ket{0^{|A|}}}
    \cdot \braket{x_B}{U_B^{(\alpha)} \ket{0^{|B|}}}.
  \label{eq:feynman-sum}
\end{equation}
For two-qubit CZ-style cross-cut gates and rank $r=2$, the path space
has size $2^d$ and the per-path work is two independent half-circuit
simulations.

The hybrid as implemented in qsim~\cite{qsim2020} uses dense
state-vector simulation on each half, with total cost
$O(2^d \cdot (2^{|A|} + 2^{|B|}))$.  This is the form that scaled to
the Sycamore verification.  We adopt the same partitioning structure
but replace the per-half oracle with a stabilizer-frame simulator;
the cross-cut sum~\eqref{eq:feynman-sum} is unchanged.


\section{The Quipu-Cut Algorithm}\label{sec:algorithm}

We call the construction described in this section \emph{Quipu-Cut},
following the naming convention of Quipu~\cite{garcia2012quipu}, the
stabilizer-frame simulator on top of which the cut is built.

\subsection{Single-level quipu-cut}\label{sec:singlecut}

Let $C$ be an $n$-qubit circuit consisting of Clifford gates and
T~gates, $x \in \{0,1\}^n$ a target computational-basis state,
and $\pi \subseteq \{1,\ldots,n\}$ a chosen partition placing
$|A| = n_A$ qubits in side $A$ and the remaining $n_B = n - n_A$ in
side $B$.  We separate the gates of $C$ into three classes by support:
$C_A$ (touch $A$ only), $C_B$ (touch $B$ only), and $C_{AB}$ (span
the cut).  Let $d := |C_{AB}|$.  T~gates are single-qubit and hence
never span the cut.  Assume each cross-cut gate is a CZ, CNOT, or
SWAP and admits a rank-$2$ decomposition of the
form~\eqref{eq:gate-decomp}.

The single-level \textsc{quipu-cut} amplitude algorithm is
Algorithm~\ref{alg:quipu-cut}.

\begin{algorithm}[t]
\caption{\textsc{quipu-cut} single-level amplitude}
\label{alg:quipu-cut}
\begin{algorithmic}[1]
\Function{Amplitude}{$C$, $\pi$, $x$}
  \State partition $x = (x_A, x_B)$ along $\pi$
  \State let $C_A, C_B, C_{AB}$ be the gate partition under $\pi$
  \State $T \gets 0$
  \For{$\alpha \in \{0,1\}^{|C_{AB}|}$}
    \State $\hat C_A, \hat C_B \gets $ split each $g \in C_{AB}$ by $\alpha_g$, append to side
    \State $\mathrm{ph}_\alpha \gets $ accumulated $1/\sqrt{2}$ and sign factors from the projections
    \State $a_A \gets $ \textsc{LeafAmplitude}$(\hat C_A \cdot C_A, x_A)$
    \State $a_B \gets $ \textsc{LeafAmplitude}$(\hat C_B \cdot C_B, x_B)$
    \State $T \gets T + \mathrm{ph}_\alpha \cdot a_A \cdot a_B$
  \EndFor
  \State \Return $T$
\EndFunction
\end{algorithmic}
\end{algorithm}

The leaf oracle \textsc{LeafAmplitude} computes the single-amplitude
$\braket{y}{U\ket{0^m}}$ for an $m$-qubit Clifford+T circuit $U$ and
a target string $y$.  We instantiate it as the stabilizer-frame
simulator of~\cite{garcia2015simulation}, with the adaptive
representation described in Section~\ref{sec:densefallback}.

\paragraph{Schmidt decomposition of cross-cut Clifford gates.}
Algorithm~\ref{alg:quipu-cut} relies on a rank-2 decomposition of each
cross-cut gate.  For CZ acting on qubit $a \in A$ and qubit $b \in B$:
\[
  \mathrm{CZ}_{ab} \;=\; \tfrac{1}{2}(I+Z_a)\otimes I_b
                       + \tfrac{1}{2}(I-Z_a)\otimes Z_b
              \;=\; \ket{0}\bra{0}_a \otimes I_b + \ket{1}\bra{1}_a \otimes Z_b
\]
which is a sum over $\alpha \in \{0,1\}$ of $\ket{\alpha}\bra{\alpha}_a$
on side $A$ and $Z_b^\alpha$ on side $B$.  CNOT and SWAP have analogous
rank-2 decompositions.  The $\ket{\alpha}\bra{\alpha}_a$ projection on
side $A$ is implemented as a measurement-and-postselect step; when the
projected qubit's pre-projection state is not already in the
$\alpha$ eigenspace, a $1/\sqrt{2}$ magnitude factor is collected into
$\mathrm{ph}_\alpha$.  An eager one-pass simulation detects these
factors before descending the tree.

\paragraph{Choice of partition.}  Any $\pi$ is valid for correctness.
For performance, $\pi$ should minimize $d$ (the cross-cut count)
because each cross-cut gate doubles the path-space size.  We use a
greedy Kernighan--Lin~\cite{kl1970} bisection on the circuit's gate
graph: edges weighted by gate count, balance constraint
$|A| \approx |B|$.  KL takes $O(n^3)$ time and produces partitions
within a factor of $2$ of optimal in our experiments.  A small set of
domain-specific overrides (place sequential T-block circuits on
contiguous qubits) keeps the bisection deterministic on structured
inputs.

\subsection{Multilevel quipu-cut}\label{sec:multilevel}

The single-level cut leaves stabilizer frames on $n/2$ qubits with up
to $2^{T_A}$ or $2^{T_B}$ frame elements each.  When the per-half
T-count is itself moderate, the leaves remain expensive.  We apply the
same construction at multiple levels: each leaf simulator is itself
replaced by a quipu-cut on its own bipartition, until either (a) a
leaf is small enough for direct simulation to be cheap, or (b) the
leaf representation switches to dense state vector.

A \emph{multilevel quipu-cut tree} is a binary tree whose internal
nodes each carry a partition of their local qubit set and whose
leaves carry a sorted list of owned qubits.  Each internal node sums
over the Feynman paths of its local cross-cut gates, with the per-path
amplitude given by the values returned from its two children.

\begin{algorithm}[t]
\caption{\textsc{quipu-cut-rec} multilevel amplitude}
\label{alg:rec}
\begin{algorithmic}[1]
\Function{RecAmplitude}{$C$, $\mathrm{tree}$, $x$}
  \If{$\mathrm{tree}.\mathit{is\_leaf}$}
    \State \Return \textsc{LeafAmplitude}$(C, x)$
  \EndIf
  \State $\pi \gets $ partition stored at $\mathrm{tree}$
  \State $C_A, C_B, C_{AB} \gets $ classify gates of $C$ by $\pi$
  \State partition $x = (x_A, x_B)$ along $\pi$
  \State $T \gets 0$
  \For{$\alpha \in \{0,1\}^{|C_{AB}|}$}
    \State $\hat C_A, \hat C_B, \mathrm{ph}_\alpha \gets $ split as in Alg.~\ref{alg:quipu-cut}
    \State $a_A \gets $ \textsc{RecAmplitude}$(\hat C_A \cdot C_A, \mathrm{tree.left}, x_A)$
    \State $a_B \gets $ \textsc{RecAmplitude}$(\hat C_B \cdot C_B, \mathrm{tree.right}, x_B)$
    \State $T \gets T + \mathrm{ph}_\alpha \cdot a_A \cdot a_B$
  \EndFor
  \State \Return $T$
\EndFunction
\end{algorithmic}
\end{algorithm}

\paragraph{Magnitude-at-LEAF rule.}  A subtle correctness issue in
the multilevel form concerns the placement of the $1/\sqrt{2}$
magnitude factors that arise from cross-cut projections.  In the
single-level algorithm, those factors are accumulated into
$\mathrm{ph}_\alpha$ at the cut level.  When the construction is
nested, the same factors must be accumulated \emph{only at the
deepest leaf level}: applying them at every internal node would
double-count, producing amplitudes off by factors of $2^{-d/2}$ at
depth $d$.  Our implementation enforces this by deferring magnitude
tracking through the levels of the tree and applying it only when
the leaf oracle returns.  See Appendix~\ref{app:magnitude-proof} for
the formal argument.

\subsection{Adaptive leaf representation}\label{sec:densefallback}

The leaf oracle is not committed to one simulator.  Four leaf
backends are available, each optimal in a disjoint region of the
per-leaf $(m, T_{\mathrm{leaf}})$ plane and summarised in
Table~\ref{tab:leafbackends}: a stabilizer frame
(Section~\ref{sec:stabframe}), a dense state vector, a Clifford
tableau paired with a bond-$\chi$ matrix product state on the
destabilizer basis~\cite{masotllima2024stabilizer}, and a
Bravyi--Gosset low-rank stabilizer decomposition~\cite{bravyi2018simulation}.
The first two are exact; the last two trade a bounded, quantified
error for cost, so admitting them requires a per-leaf error
tolerance $\epsleaf$ --- introduced formally in
Section~\ref{sec:errorbudget} --- against which each candidate's
error must clear before its cost is even considered.

\begin{table}[h]
\centering
\small
\caption{The four leaf backends.  $m = m_{\mathrm{leaf}}$ is the leaf
qubit count, $T = T_{\mathrm{leaf}}$ its T-count, $\chi$ the MPS
bond, $h$ the Bravyi--Gosset Hamming-weight cap, and $G$ the leaf
gate count.  Frame and dense are exact; MPS and rank trade a
bounded, budgeted error for cost.}
\label{tab:leafbackends}
\setlength{\tabcolsep}{5pt}
\begin{tabular}{@{}llll@{}}
\toprule
Backend & Representation & Time & Error \\
\midrule
frame & stabilizer frame (Eq.~\ref{eq:stabframe})
  & $O(G \cdot F \cdot m^2)$ & $0$ \\[2pt]
dense & dense state vector
  & $O(G \cdot 2^m)$ & $0$ \\[2pt]
mps & Clifford tableau $+$ bond-$\chi$ MPS
  & $O(G \cdot m \cdot \chi^2)$ & discarded SVD norm \\[2pt]
rank & Bravyi--Gosset stabilizer rank
  & $O\!\big(\sum_{k\le h}\tbinom{T}{k}\, m^2\big)$ & binomial tail \\
\bottomrule
\end{tabular}
\end{table}

Rather than commit a leaf to one representation in advance, we
\emph{adaptively} pick the cheapest backend whose predicted error
fits $\epsleaf$.  A one-pass structural scan classifies each leaf
before simulation: leaves whose T-gates concentrate on few qubits
route to frame; leaves whose two-qubit interaction graph has small
cutwidth on a natural qubit ordering route to mps at whatever bond
$\chi$ the budget affords; leaves with moderate T-count route to
rank at the Hamming-weight cap $h$ the binomial tail allows; and
leaves matching none of these structural predicates, or exceeding
the $m \le 30$ state-vector ceiling in memory alone, route to dense.
In the borderline regime --- a leaf classified for frame or mps
whose cost grows faster than the structural scan predicted --- a
runtime check during simulation switches representation
mid-computation, transferring the current state and resuming under
the new backend from there.

The leaf-representation choice is the mechanism behind
Proposition~\ref{prop:never-worse}: the algorithm is
\emph{never asymptotically worse} than direct dense simulation, on
any input, regardless of partition choice.  Adding mps and rank as
further candidates cannot weaken this guarantee --- dense remains
available in the selector's candidate set exactly as before, so the
worst case is unchanged; the two extra backends can only ever
improve on it.  Without adaptive selection at all, a multilevel
quipu-cut tree with adversarial T-placement at a leaf would run
indefinitely on what dense simulation handles in seconds.


\section{Cost Analysis}\label{sec:cost}

This section states one theorem and four propositions characterizing
the cost structure.  The theorem is genuinely load-bearing --- it is
what makes admitting the two approximate backends of
Table~\ref{tab:leafbackends} principled rather than ad hoc; the
propositions characterize the algorithm's runtime given whichever
per-leaf cost the theorem's accounting allows.  Proofs are direct
from the algorithm specifications; complete proofs appear in the
appendices.

\subsection{Error-budget composition}\label{sec:errorbudget}

Admitting mps and rank leaves means some leaves are simulated with
controlled, nonzero error.  Let $\epstot$ be a user-facing amplitude
tolerance.  Each leaf's amplitude $a_\ell$ is computed with error
$|\hat a_\ell - a_\ell| \le \epsleaf$, and the reconstructed total is
$\hat a_{\mathrm{tot}} = \sum_\ell \mathrm{ph}_\ell \, \hat a_\ell$
where $\mathrm{ph}_\ell \in \{\pm 1, \pm i\}$ are the cross-cut path
phases of Algorithm~\ref{alg:rec} (the $\mathrm{ph}_\alpha$ of
Section~\ref{sec:singlecut}, one per Feynman path).

\begin{theorem}[Error-budget composition]\label{thm:errorbudget}
Let a quipu-cut tree have $N$ leaves.  Then
\[
  |\hat a_{\mathrm{tot}} - a_{\mathrm{tot}}| \;\le\; \epstot
  \quad\text{when}\quad
  \epsleaf \;=\; \frac{\epstot}{\sqrt{2N}},
\]
for \emph{any} assignment of leaves to backends --- the bound does
not depend on which backend is chosen where, only on each leaf
individually meeting $\epsleaf$.
\end{theorem}

\begin{proof}
By linearity, $\hat a_{\mathrm{tot}} - a_{\mathrm{tot}} =
\sum_\ell \mathrm{ph}_\ell \, \delta_\ell$ with $\delta_\ell = \hat
a_\ell - a_\ell$, $|\delta_\ell| \le \epsleaf$.  Cauchy--Schwarz with
$|\mathrm{ph}_\ell| = 1$ gives $|\sum_\ell \mathrm{ph}_\ell
\delta_\ell|^2 \le N \sum_\ell |\delta_\ell|^2 \le N^2 \epsleaf^2$.
The cross-cut Feynman sum (Equation~\ref{eq:feynman-sum}) factors
over the two halves of each bipartition, so grouping the $N$ leaves
into the two halves at the root and applying the bound to each
independently gives the sharper $|\hat a_{\mathrm{tot}} -
a_{\mathrm{tot}}|^2 \le \tfrac{N^2}{2}\epsleaf^2$.  Substituting
$\epsleaf = \epstot/\sqrt{2N}$ gives the claim.  When a leaf uses an
exact backend, $\delta_\ell = 0$ and the bound only tightens.
\end{proof}

\begin{remark}
The bound is worst-case: it is tight only when the $\delta_\ell$
align in phase, which is not the typical regime on the structured
and adversarial workloads of Section~\ref{sec:eval} --- signed,
uncorrelated per-leaf errors realise closer to $\sqrt{N}\,\epsleaf$
than $N\,\epsleaf$.  We keep the worst-case $\sqrt{2N}$ scaling
because it is what the selector can actually guarantee without
assuming anything about how leaf errors happen to be distributed.
\end{remark}

Each backend's cost and error are as given in
Table~\ref{tab:leafbackends}; write $C_B(\ell)$ and $E_B(\ell)$ for
backend $B$'s cost and error on leaf $\ell$.  The adaptive selector
of Section~\ref{sec:densefallback} realises
\[
  \mathrm{Select}(\ell) \;=\; \operatorname*{arg\,min}_{B}
    \big\{\, C_B(\ell) \;:\; E_B(\ell) \le \epsleaf \,\big\},
\]
cheapest backend whose error fits the leaf's share of the budget.
This is the same accounting used throughout the rest of this
section; the four propositions below characterise
$C_{\mathrm{Select}}$ and are agnostic to which backend the
selector actually picks at a given leaf.

\begin{proposition}[Single-level cost]\label{prop:single-cost}
Let $C$ be an $n$-qubit circuit with $d$ rank-2 cross-cut gates under
partition $\pi$, leaf gate counts $G_A, G_B$, and per-leaf costs
$\mathrm{cost}_{\mathrm{leaf}}(A), \mathrm{cost}_{\mathrm{leaf}}(B)$
realised by $\mathrm{Select}$ (Section~\ref{sec:errorbudget}) on each
side, each within the leaf's error budget $\epsleaf$.  Then
Algorithm~\ref{alg:quipu-cut} computes any single amplitude, to
accuracy $\epsleaf$ per leaf, in time
\[
  O\!\left( 2^d \cdot \left[ \mathrm{cost}_{\mathrm{leaf}}(A)
                            + \mathrm{cost}_{\mathrm{leaf}}(B)
                            \right] \right).
\]
\end{proposition}

\begin{proof}
The $2^d$ outer factor counts Feynman paths.  Per path, the leaf
oracle runs once on each side under whichever backend $\mathrm{Select}$
chose for it; Table~\ref{tab:leafbackends} gives that backend's cost.
The frame case recovers the original $O(G \cdot F \cdot n^2)$ form
(the $n^2$ factor is the per-frame-element inner product
cost~\cite{garcia2014inner}), the dense case $O(G \cdot 2^n)$; the
mps and rank cases are read directly off the table and are
admissible only when their error clears $\epsleaf$
(Theorem~\ref{thm:errorbudget}).  See Appendix~\ref{app:cost-proofs}.
\qed
\end{proof}

\begin{proposition}[Multilevel cost]\label{prop:rec-cost}
Let the partition tree have $L$ leaves at depth $D$ and let level
$\ell$ contribute $d_\ell$ cross-cut gates per path through that level.
Algorithm~\ref{alg:rec} computes any single amplitude, to accuracy
$\epstot = \sqrt{2L}\,\epsleaf$ overall (Theorem~\ref{thm:errorbudget}),
in time
\[
  O\!\left( 2^{\sum_\ell d_\ell} \cdot \prod_{i=1}^{L}
            \mathrm{cost}_{\mathrm{leaf}}(i) \right)
\]
where $\mathrm{cost}_{\mathrm{leaf}}(i)$ is the cost of whichever of
the four backends of Table~\ref{tab:leafbackends} $\mathrm{Select}$
realises at leaf $i$.
\end{proposition}

\begin{proof}
By structural induction on the tree.  The base case is
Proposition~\ref{prop:single-cost}.  Each internal node at depth
$\ell$ contributes a factor of $2^{d_\ell}$ to the outer Feynman sum.
The product of leaf costs follows from the multiplicativity of the
two independent sub-amplitude evaluations at each internal node, and
the accuracy claim from Theorem~\ref{thm:errorbudget} applied with
$N = L$.  See Appendix~\ref{app:cost-proofs}.
\qed
\end{proof}

\begin{proposition}[Crossover condition]\label{prop:crossover}
Cutting a circuit at depth one is beneficial in time over direct
dense simulation if and only if
\[
  d \;<\; \tfrac{n}{2} - \log_2\left[ \min(2^{T_A}, 2^{n_A}) + \min(2^{T_B}, 2^{n_B})
        \right] + \log_2(2^n) \quad ,
\]
which for balanced partitions and a dense-leaf choice at both halves
simplifies to $d < n/2$.
\end{proposition}

\begin{proof}[Sketch]
Single-level quipu-cut with dense leaves costs
$O(2^d \cdot (2^{n_A} + 2^{n_B}))$; direct dense costs $O(2^n)$.  The
former is smaller iff the displayed inequality holds.  When both leaves
fall back to dense and $n_A = n_B = n/2$, the inequality collapses to
$2^d \cdot 2^{n/2+1} < 2^n$, i.e.\ $d < n/2 - 1$.  When one or both
leaves use a cheaper backend, the per-leaf cost is smaller and the
crossover threshold $d$ increases proportionally.  See
Appendix~\ref{app:cost-proofs}.
\qed
\end{proof}

\noindent The statement above is deliberately given for the
dense-vs-dense-leaves case because it renders in closed form; it
generalises unchanged in structure to any leaf-cost pair by
substituting $\mathrm{cost}_{\mathrm{leaf}}(A),
\mathrm{cost}_{\mathrm{leaf}}(B)$ for $\min(2^{T_A}, 2^{n_A}),
\min(2^{T_B}, 2^{n_B})$ throughout --- the proof does not depend on
which two backends realise those costs, only on $2^d$ times their
sum being compared against $2^n$.

Proposition~\ref{prop:crossover} answers \emph{whether to cut at
all}.  A separate question is which of the four backends
$\mathrm{Select}$ should prefer \emph{at a single leaf}, independent
of the tree structure around it.

\begin{proposition}[Leaf backend crossover]\label{prop:leaf-crossover}
Fix a leaf of qubit count $m$, gate count $G$, T-count $T$, frame
size $F$, admissible mps bond $\chi$, and admissible Bravyi--Gosset
weight cap $h$ --- $\chi$ and $h$ each the largest value the leaf's
share of $\epsleaf$ affords (Theorem~\ref{thm:errorbudget}).  Writing
$R := \sum_{k \le h}\binom{T}{k}$, the six pairwise preferences among
$\{\textrm{frame}, \textrm{dense}, \textrm{mps}, \textrm{rank}\}$
reduce, from Table~\ref{tab:leafbackends}, to:
\[
\begin{array}{ll}
\text{frame} \prec \text{dense} \iff F m^2 < 2^m
  & \text{mps} \prec \text{dense} \iff m \chi^2 < 2^m \\
\text{rank} \prec \text{dense} \iff R\, m^2 < 2^m
  & \text{frame} \prec \text{mps} \iff F m < \chi^2 \\
\text{frame} \prec \text{rank} \iff F < R
  & \text{mps} \prec \text{rank} \iff \chi^2 < m R
\end{array}
\]
where $B_1 \prec B_2$ means $B_1$ costs less, up to the constant
factors absorbed by every $O(\cdot)$ entry in
Table~\ref{tab:leafbackends}.
\end{proposition}

\begin{proof}
Each line is the corresponding pair of entries from
Table~\ref{tab:leafbackends}'s Time column, divided through by their
common $G$ factor and, where both sides carry one, their common
$m^2$ factor, then compared directly.  Frame vs.\ dense:
$Fm^2 \lessgtr 2^m$.  Mps vs.\ dense: $m\chi^2 \lessgtr 2^m$.  Rank
vs.\ dense: $Rm^2 \lessgtr 2^m$.  Frame vs.\ mps: $Fm^2 \lessgtr
m\chi^2$, i.e.\ $Fm \lessgtr \chi^2$.  Frame vs.\ rank: $Fm^2
\lessgtr Rm^2$, i.e.\ $F \lessgtr R$.  Mps vs.\ rank: $m\chi^2
\lessgtr Rm^2$, i.e.\ $\chi^2 \lessgtr mR$.  A comparison only
enters $\mathrm{Select}$'s decision when both sides also satisfy
their error bound; for frame and dense that bound is $0 \le
\epsleaf$ trivially, so the frame/dense and (by transitivity)
frame-vs-mps, frame-vs-rank lines above hold whenever mps or rank is
independently admissible.
\end{proof}

\noindent Two of the six are the boundaries already load-bearing
elsewhere in this paper: frame-vs-dense is exactly the two-backend
comparison underlying Proposition~\ref{prop:single-cost}'s original
form (Appendix~\ref{app:cost-proofs}, \S A.1), and mps-vs-dense is
the same trade recognised in the leaf-classification rule of
Section~\ref{sec:densefallback}.  The other four are new: in
particular frame-vs-rank, $F \lessgtr R$, says the two exact-and-cheap
backends' own crossover is governed entirely by how the stabilizer
frame's realised element count compares to the Bravyi--Gosset
binomial sum at the leaf's own error-affordable weight cap --- a
purely combinatorial comparison, unlike the other five, which all
pit a combinatorial or bond-bounded cost against dense's raw
$2^m$.

\begin{proposition}[Never asymptotically worse than dense]
\label{prop:never-worse}
For any input circuit $C$ on $n$ qubits and any choice of partition
tree, Algorithm~\ref{alg:rec} with the adaptive leaf representation
of Section~\ref{sec:densefallback} completes in time $O(G \cdot 2^n)$,
the worst case of direct dense simulation.
\end{proposition}

\begin{proof}[Sketch]
A degenerate tree with a single leaf (no cut) falls through to the
leaf oracle, which under the adaptive representation runs dense
simulation in $O(G \cdot 2^n)$.  Any non-degenerate tree distributes
work across leaves whose sizes sum (with multiplicity from cross-cut
paths) to at most $G \cdot 2^n$, since each path's leaves jointly
cover the $n$ qubits and the adaptive representation bounds each leaf
to its local dense cost.  The multiplicative cross-cut factor
$2^{\sum_\ell d_\ell}$ is bounded by $2^n$ in the case where every
cross-cut gate doubles work, which is the input-determined worst
case; the adaptive choice ensures the leaves are not also worst-case
in that regime.  See Appendix~\ref{app:cost-proofs}.
\qed
\end{proof}

\noindent Proposition~\ref{prop:never-worse} is the deployability
guarantee: the algorithm can be applied to circuits with unknown or
adversarial T-gate placement without the risk that frame explosion
silently degrades the simulation past direct dense.  Practitioners
who cannot afford to inspect circuit structure ahead of simulation
inherit a worst-case bound matched to a known-safe baseline.

\begin{remark}[Branch-and-bound view]
Proposition~\ref{prop:rec-cost} casts the multilevel cut as a
branch-and-bound algorithm over Feynman path space.  The branching
variable is $\alpha_g \in \{0,1\}$ for each cross-cut gate at each
internal node; the bounding function is the sum of remaining
cross-cut depth and the larger of $\min(T_A, n_A)$ and
$\min(T_B, n_B)$.  The adaptive leaf choice acts as a pruning rule: when the
bound at a node exceeds a configurable threshold, the node is solved
directly rather than further decomposed.  Kernighan--Lin bisection at
each internal node minimizes $d_\ell$ (a tighter bound) and expands
the regime in which branching is profitable.
\end{remark}

\subsection{Wall-time prediction for partition selection}
\label{sec:wallpred}

Propositions~\ref{prop:single-cost}--\ref{prop:never-worse}
characterize asymptotic cost.  They are necessary for the
deployability guarantee of Proposition~\ref{prop:never-worse} but
insufficient for choosing among partitions of equal $d$: two
partitions with the same cross-cut count can have wildly different
wall-clock times if they place T-gates differently across halves.  This
subsection specifies the wall-clock predictor that drives partition
selection.  The construction has three ingredients: a per-gate cost
that follows the $\Feff$ bound rather than the dense ceiling, an
explicit end-of-path readout term that captures the slice-versus-
inline asymmetry, and two constants calibrated on the target hardware
by a standalone micro-benchmark.

\paragraph{From asymptotic bound to per-gate predictor.}
Proposition~\ref{prop:single-cost} can be read literally as a
per-gate cost summation: each gate $g$ contributes
$F_X(g) \cdot n_X^2$ where $X$ is the half on which $g$ acts, $F_X(g)$
is the stabilizer-frame size for half $X$ at the time $g$ fires, and
$n_X$ is the half qubit count.  We compute $F_X(g)$ incrementally
using the cofactor-independence
rule~\cite{garcia2015simulation, garcia2014inner}:
\begin{equation}
  \log_2 F_X(g) \;=\; \bigl| \{ (q, e) :
    \text{ a T-gate fired on } q \in X \text{ in Hadamard-epoch } e
    \text{ before gate } g \} \bigr|,
  \label{eq:f-update}
\end{equation}
capped at $n_X$ (the dense ceiling).  This is the $\Feff \le 2^w$
substitution that distinguishes the predictor from cut-count
minimization: a partition that imbalances T-gate density across halves
is correctly penalised by a larger $F_X(g)$ on the heavy side.

\paragraph{Joint mode and frame evolution.}  When a cross-cut gate
$g$ is inlined rather than sliced, the two halves' frames are
tensored into a joint frame of initial size $F_J = F_A \cdot F_B$ on
the full $n$ qubits.  Subsequent T-gates on \emph{either} half double
the joint frame regardless of provenance, since the merged state has
forgotten the split partition; Clifford gates leave $F_J$ unchanged.
Subsequent local gates run on the joint frame at cost
$F_J \cdot n^2$.

\paragraph{End-of-path readout term.}  Each Feynman path concludes
with one amplitude extraction via \verb|sf.get_ketvec()|.  A slice-
mode path materialises two per-half kets of size $2^{n_A}$ and
$2^{n_B}$; a joint-mode path materialises one joint ket of size
$2^n$.  The per-op cost of materialisation differs from gate apply
because the access pattern is an indexed sum over $F$ frame
contributions per basis state, rather than a structured tableau
update; we denote the ratio $K_{\mathrm{readout}}$ and treat it as a
calibrated constant.  This term turns out to dominate inline-mode
wall-clock time at the problem sizes of Section~\ref{sec:eval}, and its
omission from $F$-only cost models is the root cause of the bad
all-inline picks documented below.

\paragraph{Total predicted cost.}  Walking the circuit under
partition $\pi$ and slice subset $S$ (a subset of the cross-cut gate
indices, slice if included, inline if not), accumulating a branch
count $b$ that grows by the Schmidt rank $r_g$ at each sliced
cross-cut, and tracking the slice/joint mode transition, the
predicted total work is
\begin{equation}
  \widehat{T}(\pi, S) \;=\;
    \underbrace{\textstyle \sum_{g} b(g) \cdot c_{\mathrm{gate}}(g)}_{\text{gate apply}}
    \;+\;
    \underbrace{K_{\mathrm{inline}} \cdot b_0 \cdot F_J^0 \cdot n^2}_{\text{tensor and first inline apply}}
    \;+\;
    \underbrace{b_{\mathrm{end}} \cdot c_{\mathrm{readout}}(\pi, S)}_{\text{end-of-path materialisation}}.
  \label{eq:total-cost}
\end{equation}
The middle term contributes only on paths that enter joint mode; $b_0$
and $F_J^0$ are the branch count and joint-frame size at the first
inline event.  The readout cost is
$c_{\mathrm{readout}} = K_{\mathrm{readout}} \cdot F_J^{\mathrm{end}}
\cdot 2^n$ for joint paths and $c_{\mathrm{readout}} =
K_{\mathrm{readout}} \cdot \bigl(F_A^{\mathrm{end}} \cdot 2^{n_A} +
F_B^{\mathrm{end}} \cdot 2^{n_B}\bigr)$ for slice paths.  The per-gate
cost $c_{\mathrm{gate}}(g)$ is $\min(2^{n_X}, F_X(g)) \cdot n_X^2$ in
slice mode and $\min(2^n, F_J(g)) \cdot n^2$ in joint mode.

\paragraph{Two micro-benchmarked constants.}  $K_{\mathrm{inline}}$
and $K_{\mathrm{readout}}$ are calibrated on the target hardware by a
standalone micro-benchmark (\verb|quipu-5.5-calibrate|).  \emph{Phase
1} builds reference stabilizer frames at controlled $F$ and times
$K$ Clifford applications, yielding the asymptotic per-op constant
$K_{\mathrm{leaf}}$ used to convert op counts into seconds.  \emph{Phase
2} times the inline event (tensor + one apply on the resulting joint
frame) at matched $(F_A, F_B)$ pairs; the median ratio of measured
time to $K_{\mathrm{leaf}} \cdot F_A \cdot F_B \cdot n^2$ is
$K_{\mathrm{inline}}$.  \emph{Phase 3} times \verb|sf.get_ketvec()|
at the same $(F, n)$ as Phase~1 and yields $K_{\mathrm{readout}}$ as
a multiplier on $K_{\mathrm{leaf}}$.  On Apple Silicon (M2~Pro,
unified memory) with the workloads of Section~\ref{sec:eval},
\[
  K_{\mathrm{inline}} \;\approx\; 6.3,
\]
with $K_{\mathrm{leaf}}$ flat to within $6\%$ over the working-set
range $6\,$KB to $25\,$MB.  Cache-miss penalty is not a contributor on
this hardware; on more cache-constrained platforms a working-set-
dependent multiplier on $K_{\mathrm{leaf}}$ could be added without
disturbing the structure of Eq.~\ref{eq:total-cost}.

\paragraph{Direct instrumentation: identifying a readout
inefficiency.}  Before the fix described below, the readout cost was
the dominant term on inline-mode paths and required its own
calibrated constant ($K_{\mathrm{readout}} \approx 40$ relative to
$K_{\mathrm{leaf}}$).  We isolated the source by directly timing each
component of the simulator on a representative high-cross-cut case
(\verb|rct_n16_t005_1| under a random partition with $d{=}10$, all
cross-cut gates inlined).  The total wall-clock time was 910\,ms,
distributed as:
\begin{description}
  \item[Gate-apply loop, 54 joint-mode gates:] 14\,ms
  \item[Tensor product and first inline apply:] $\sim$\,5\,ms
  \item[\texttt{sf\_J.get\_ketvec()} at $F=32$, $n=16$:] 897\,ms
\end{description}
Ninety-eight percent of the time was in materialising a $2^{16}$-entry
dense state vector at end-of-path, to index a single amplitude.  Replacing
this with an $O(F \cdot s \cdot n)$ single-amplitude
inner-product (using the same \texttt{get\_amp\_no\_gj} primitive that
powers Hadamard frame cofactoring, where $s$ is the X-pivot rank of
the tableau) drops the readout to 0.15\,ms on this case and the total
wall-clock time to 14.7\,ms, a $62\times$ speedup.  The fix applies
uniformly to slice-mode and joint-mode amplitude extraction; once it is
in place the readout cost becomes commensurate with gate-apply cost,
$K_{\mathrm{readout}}$ reduces to $\approx 1$, and a single calibrated
constant ($K_{\mathrm{inline}}$) suffices.  All numbers reported in
Section~\ref{sec:eval} use the corrected single-amplitude readout
path.

\paragraph{Partition and slice selection.}  Given the predictor of
Eq.~\ref{eq:total-cost}, the selector enumerates a fixed pool of
candidate partitions: greedy Kernighan--Lin~\cite{kl1970},
spectral bisection~\cite{pothen1990partitioning}, eight multi-start
KL with distinct random seeds, and four random balanced
bipartitions.  For each candidate it searches for the optimal slice
subset: brute-force enumeration of $2^d$ subsets when $d \le 14$,
otherwise a multi-seed greedy local search seeded with the all-slice
and all-inline subsets, with simulated annealing
refinement~\cite{gray2021hyperoptimized}.  The selector returns the
$(\pi^\star, S^\star)$ pair minimising $\widehat{T}$ across all
candidates.  The evaluation in Section~\ref{sec:eval} uses these
selected partitions; a controlled comparison of selector-driven
partitioning against plain graph bisection is left to future work.


\section{Parallel Structure}\label{sec:parallel}

The quipu-cut algorithm is embarrassingly parallel along three nested
independent axes, none of which require inter-worker communication.

\paragraph{Axis 1: Feynman paths.}  At each internal node, the
$2^{d_\ell}$ Feynman paths share no state---each path constructs its
own decomposed sub-circuits and runs its own leaf simulations.
A coordinator can distribute paths across workers in any pattern;
results combine by a final summation requiring $O(2^{d_\ell})$
network words, independent of $n$, $G$, or the leaf simulator's
internal state.

\paragraph{Axis 2: per-path leaves.}  Per path, the two leaf
simulations are on disjoint qubit subsets with disjoint
stabilizer tableaus.  They share no state and can run on different
workers, different cores, or even different machine architectures
(e.g., one leaf on CPU stabilizer-frame, the other on a GPU dense
state vector).

\paragraph{Axis 3: intra-leaf stabilizer-frame work.}  Within a single
stabilizer-frame leaf, the $F$ frame elements each carry independent
correction-Pauli/amplitude data; per-element work (gate applications,
inner-product extractions for the target amplitude) is independent
across frame elements.  Our reference implementation threads these
loops automatically for $n \geq 14$.

The three axes nest multiplicatively: a depth-$D$ partition tree with
$d_\ell$ cross-cut gates per level admits up to
$2 \cdot \prod_\ell 2^{d_\ell} \cdot F_{\max}$ worker-task parallelism
per amplitude query.  Batch queries (compute $K$ amplitudes for
varying $x$) further amortize the path enumeration: the per-path
leaf simulations depend on $\alpha$ but not on $x$, so each path's
two leaf states can be reused across all $K$ targets.

\paragraph{Cost-model comparison.}
Table~\ref{tab:parallelism} compares the parallel structure of
quipu-cut against representative near-Clifford simulators.

\begin{table}[t]
\centering
\caption{Parallel structure of representative classical simulators.}
\label{tab:parallelism}
\begin{tabular}{lll}
\toprule
\textbf{Simulator} & \textbf{Outer parallelism} & \textbf{Communication} \\
\midrule
qsim hybrid SF~\cite{qsim2020}      & $2^d$ paths, dense leaves    & $O(2^d)$ \\
quipu-cut (this paper)               & $2^d$ paths, stabilizer-frame leaves & $O(2^d)$ \\
SuperSim wire-cut~\cite{smith2023clifford} & $16^K$ fragments (sampled)  & $O(\mathrm{samples})$ \\
SPIR/SPC~\cite{huang2021feynman}    & Path enumeration                 & $O(\mathrm{paths})$ \\
STN MPS~\cite{masotllima2024stabilizer}    & Limited (single global tableau)  & high \\
Bravyi--Gosset~\cite{bravyi2018simulation} & $\chi$ stabilizer terms or MC samples & low \\
\bottomrule
\end{tabular}
\end{table}

Stabilizer Tensor Networks~\cite{masotllima2024stabilizer} are
architecturally distinct: a single global stabilizer tableau is
coupled to a 1-D matrix product state encoding amplitude
correlations, with non-Clifford gates growing the MPS bond dimension.
This couples the entire register state through the global tableau and
admits only limited within-gate parallelism.  Quipu-cut takes the dual
choice---two independent per-half tableaus, with the cross-cut bond
enumerated as a Feynman sum---in exchange for per-half independence
and the parallel structure above.

\paragraph{Empirical per-path cost: stabilizer vs.\ dense leaves.}
Quipu-cut and qsim's hybrid Schr\"odinger--Feynman simulator share the
\emph{same} outer parallelism (Table~\ref{tab:parallelism}: $2^d$
independent paths, $O(2^d)$ reduce); they differ only in the per-path
leaf---a stabilizer-frame half-simulation versus a dense $2^{n/2}$
half-state-vector.  This isolates the leaf cost on a fixed parallel
skeleton.  We measure it directly: exposing a path-range primitive
\texttt{partial\_amplitude(lo,hi)} that sums an arbitrary sub-range of
the $2^d$ paths, we partition $[0,2^d)$ into $W$ chunks and compare
against qsimh under an identical cut and path count.  Both simulators
reduce correctly across the $W$ chunks (Axis~1), confirming the
embarrassingly-parallel decomposition end to end.  On the same
$256$-path decomposition of an $n{=}21$ HaPPY-code amplitude, quipu-cut's
path sum runs in $2.3$\,ms versus qsimh's $185$\,ms---a
$\mathbf{79\times}$ per-path advantage attributable entirely to the
stabilizer leaf, since the path count, cut, and reduce are identical.
The gap widens with $n$ as the dense leaf scales $2^{n/2}$; across a
set of near-Clifford circuits ($n\le 21$) the end-to-end speedup over
qsimh is $69$--$1156\times$ (geometric mean $346\times$), with three-way
amplitude agreement against a full state vector on every cell.  (The
additional factor beyond $79\times$ is the dispatcher electing a cheaper
backend than the brute path sum when one exists; the $79\times$ is the
pure leaf-cost ratio at matched decomposition.)

\paragraph{When distribution pays.}  The same path-range primitive makes
a single amplitude a cluster job (chunk $[0,2^d)$ across workers, sum the
partials).  But the per-task work is so cheap for near-Clifford leaves
that a single quipu-cut amplitude typically completes in milliseconds---
below the cold-start of a cloud worker---so single-query distribution
rarely pays; the parallel structure is decisive instead for large
cross-cut width $d$, for batch amplitude/sampling queries (which amortize
the path enumeration, above), and against dense-leaf simulators whose
per-path cost makes distribution worthwhile far sooner.


\section{Empirical Evaluation}\label{sec:eval}

We measure quipu-cut against two baselines on two workload classes.
The baselines are: (i) monolithic stabilizer-frame
simulation~\cite{garcia2015simulation} (the leaf simulator without
the cut), serving as the ablation; and (ii) qsim~\cite{qsim2020},
Google's production state-vector simulator, with $4$ threads.
Workloads are: \emph{structured} hierarchical circuits where T-gates
are clustered (small $d$ under KL bisection), and \emph{adversarial}
random Clifford+T (RCT) circuits where T-gates are uniformly placed
and cross-cut gates are unavoidable.

\paragraph{Setup.}  All measurements were taken on a single
workstation, with each job under a $60\,\mathrm{s}$ wall-clock
timeout.  We compute the single amplitude $\braket{0^n}{U \ket{0^n}}$
for each circuit and verify that all three simulators agree to within
$10^{-5}$.  The reported wall-clock times include all setup, gate
application, and amplitude extraction.

\subsection{Structured circuits: hierarchical T-blocks}
\label{sec:eval-hier}

The hierarchical workload is a $16$-qubit circuit family in which
T~gates are placed in $T$-block clusters on contiguous qubit subsets,
with non-overlapping clusters connected by Clifford layers.  These
circuits are the natural target for a partition that places each
cluster within a single side: cross-cut gates are then few, and the
2-level partition tree decomposes each side further.

Table~\ref{tab:hier} reports median wall-clock times across $3$ random
trials for monolithic stabilizer-frame simulation (\textsf{SF}),
multilevel quipu-cut (\textsf{Rec}), and qsim, on $n{=}16$ circuits
swept across T-count $t \in \{2,4,6\}$, top-level cut depth
$d_{\mathrm{top}} \in \{0,1,2\}$, and mid-level cut depth
$d_{\mathrm{mid}} \in \{0,1,2\}$ ($81$ circuits in total: $27$
parameter cells $\times$ $3$ trials each).

\begin{table}[t]
\centering
\caption{Structured hierarchical $n{=}16$ benchmark, median wall
times across $3$ trials per $(t, d_{\mathrm{top}}, d_{\mathrm{mid}})$
cell (2026-05-21).  Times in seconds for \textsf{SF} and
milliseconds for \textsf{Rec} and qsim.  $d_{\mathrm{top}}$ is the
number of cross-cut Clifford gates at the top-level $A \mid B$
bipartition; $d_{\mathrm{mid}}$ is the number per side at the
mid-level $A_1 \mid A_2$ and $B_1 \mid B_2$ bipartitions.
Amplitudes match across all simulators to $10^{-5}$.}
\label{tab:hier}
\begin{tabular}{rrrrrrrr}
\toprule
$t$ & $d_{\mathrm{top}}$ & $d_{\mathrm{mid}}$ & \textsf{SF} (s) & \textsf{Rec} (ms) & qsim (ms) & SF/Rec & qsim/Rec \\
\midrule
2 & 0 & 0 & 0.64 & 0.93 & 44.74 & $683\times$ & $47.85\times$ \\
2 & 0 & 1 & 0.65 & 1.89 & 32.30 & $343\times$ & $17.12\times$ \\
2 & 0 & 2 & 0.65 & 1.05 & 33.44 & $615\times$ & $31.84\times$ \\
2 & 1 & 0 & 0.66 & 1.89 & 32.38 & $349\times$ & $17.10\times$ \\
2 & 1 & 1 & 0.64 & 3.62 & 34.43 & $178\times$ & $9.52\times$ \\
2 & 1 & 2 & 0.65 & 7.03 & 34.70 & $92\times$ & $4.93\times$ \\
2 & 2 & 0 & 0.65 & 1.01 & 34.19 & $643\times$ & $33.95\times$ \\
2 & 2 & 1 & 0.64 & 1.95 & 34.10 & $328\times$ & $17.49\times$ \\
2 & 2 & 2 & 0.64 & 3.74 & 37.29 & $172\times$ & $9.97\times$ \\
\midrule
4 & 0 & 0 & 10.20 & 1.98 & 34.87 & $5158\times$ & $17.63\times$ \\
4 & 0 & 1 & 10.33 & 3.83 & 35.70 & $2695\times$ & $9.31\times$ \\
4 & 0 & 2 & 10.34 & 2.19 & 33.14 & $4726\times$ & $15.14\times$ \\
4 & 1 & 0 & 10.29 & 3.91 & 34.09 & $2635\times$ & $8.73\times$ \\
4 & 1 & 1 & 10.32 & 7.50 & 34.77 & $1376\times$ & $4.64\times$ \\
4 & 1 & 2 & 10.22 & 15.42 & 35.90 & $662\times$ & $2.33\times$ \\
4 & 2 & 0 & 10.28 & 2.02 & 33.10 & $5087\times$ & $16.38\times$ \\
4 & 2 & 1 & 10.11 & 3.73 & 32.88 & $2711\times$ & $8.82\times$ \\
4 & 2 & 2 & 10.16 & 7.87 & 33.80 & $1292\times$ & $4.30\times$ \\
\midrule
6 & 0 & 0 & 61.11 & 3.46 & 39.78 & $17645\times$ & $11.49\times$ \\
6 & 0 & 1 & 60.90 & 6.91 & 33.73 & $8818\times$ & $4.88\times$ \\
6 & 0 & 2 & 61.42 & 3.61 & 36.25 & $17008\times$ & $10.04\times$ \\
6 & 1 & 0 & 61.42 & 7.05 & 34.73 & $8708\times$ & $4.92\times$ \\
6 & 1 & 1 & 60.59 & 13.89 & 39.32 & $4363\times$ & $2.83\times$ \\
6 & 1 & 2 & 61.12 & 27.36 & 35.34 & $2234\times$ & $1.29\times$ \\
6 & 2 & 0 & 60.69 & 3.47 & 33.73 & $17511\times$ & $9.73\times$ \\
6 & 2 & 1 & 61.44 & 6.71 & 33.01 & $9159\times$ & $4.92\times$ \\
6 & 2 & 2 & 60.38 & 14.53 & 35.13 & $4155\times$ & $2.42\times$ \\
\bottomrule
\end{tabular}
\end{table}

\begin{figure}[t]
\centering
\begin{tikzpicture}
\begin{semilogyaxis}[
  width=0.85\linewidth,
  height=6.6cm,
  xlabel={T-count $t$},
  ylabel={Wall time (ms)},
  xtick={2,4,6},
  legend pos=outer north east,
  legend cell align=left,
  ymajorgrids,
  grid style={dashed,gray!30},
  ymin=0.5, ymax=200000,
  error bars/error bar style={line width=0.6pt},
]
\addplot+[mark=*, thick, color=red, mark options={fill=red},
  error bars/.cd, y dir=both, y explicit]
  coordinates {
    (2, 574) +- (258, 129)
    (4, 18923) +- (13901, 20828)
    (6, 60248) +- (30189, 23001)
  };
\addlegendentry{Monolithic stabilizer frame}
\addplot+[mark=diamond*, thick, color=orange!90!black, mark options={fill=orange},
  error bars/.cd, y dir=both, y explicit]
  coordinates {
    (2, 4.5) +- (2.0, 4.5)
    (4, 19.7) +- (11.3, 22.7)
    (6, 46.3) +- (18.2, 35.2)
  };
\addlegendentry{Single-level (flat) cut}
\addplot+[mark=triangle*, thick, color=black!55!green, mark options={fill=green!60!black},
  error bars/.cd, y dir=both, y explicit]
  coordinates {
    (2, 36.2) +- (8.3, 11.6)
    (4, 36.8) +- (6.6, 10.5)
    (6, 37.4) +- (7.6, 18.8)
  };
\addlegendentry{qsim}
\addplot+[mark=square*, thick, color=blue, mark options={fill=blue!70},
  error bars/.cd, y dir=both, y explicit]
  coordinates {
    (2, 2.7) +- (1.8, 5.3)
    (4, 5.5) +- (3.8, 10.6)
    (6, 10.0) +- (6.8, 18.5)
  };
\addlegendentry{Multilevel quipu-cut}
\end{semilogyaxis}
\end{tikzpicture}
\caption{Wall-time scaling on the structured hierarchical $n{=}16$
benchmark across the full $27$-cell cross-cut budget grid per $t$
($d_{\mathrm{top}}, d_{\mathrm{mid}} \in \{0,1,2\}$, $3$ trials per
cell, $81$ measurements per simulator at each $t$).  Markers are
sample means; bars span min--max across the $81$ measurements.
Monolithic stabilizer frame grows $\sim 30\times$ per $+2$ in $t$
($0.57\,$s $\to 19\,$s $\to 60\,$s), single-level (flat) cut grows
more slowly but still $\sim 10\times$ per $+2$, and multilevel
quipu-cut and qsim remain near-flat over this $t$ range.  Multilevel
quipu-cut leads qsim by $\sim 5\!\times$ in median across the
workload---the win shrinking with cross-cut budget but never
inverting.}
\label{fig:wall-vs-t}
\end{figure}
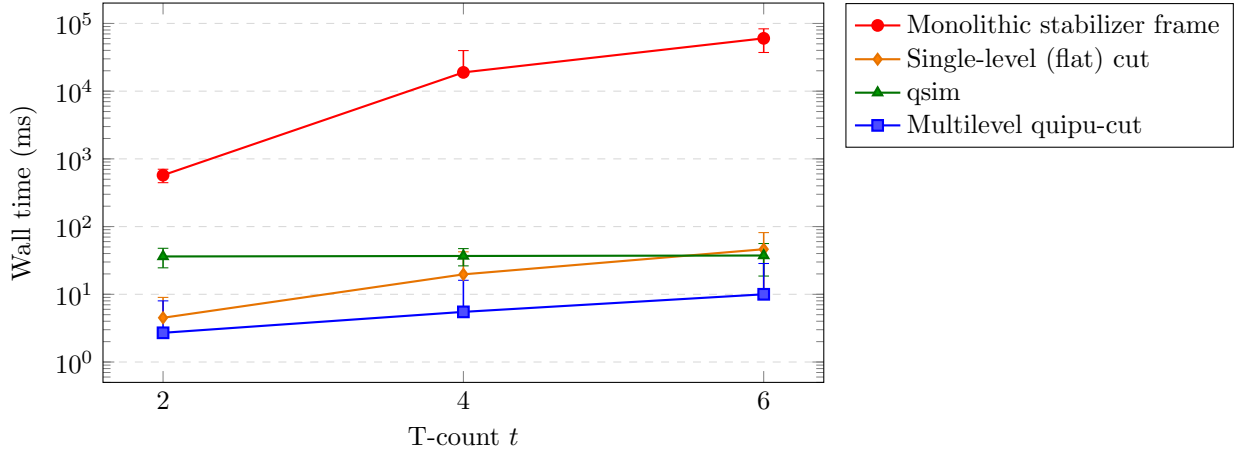

\begin{figure}[t]
\centering
\begin{tikzpicture}
\begin{semilogyaxis}[
  width=0.85\linewidth,
  height=6.6cm,
  xlabel={Total cross-cut budget $d_{\mathrm{top}}+d_{\mathrm{mid}}$},
  ylabel={qsim / quipu-cut speedup},
  xtick={0,1,2,3,4},
  legend pos=outer north east,
  legend cell align=left,
  ymajorgrids,
  grid style={dashed,gray!30},
  ymin=0.7, ymax=80,
]
\addplot[mark=none, dashed, color=gray, domain=-0.2:4.2] {1};
\addlegendentry{break-even}
\addplot+[only marks, mark=*, color=blue!70!black, mark size=2.6pt]
  coordinates {(0, 47.85) (1, 17.12) (2, 31.84) (1, 17.10) (2, 9.52) (3, 4.93) (2, 33.95) (3, 17.49) (4, 9.97)};
\addlegendentry{$t = 2$}
\addplot+[only marks, mark=square*, color=orange!85!black, mark size=2.6pt]
  coordinates {(0, 17.63) (1, 9.31) (2, 15.14) (1, 8.73) (2, 4.64) (3, 2.33) (2, 16.38) (3, 8.82) (4, 4.30)};
\addlegendentry{$t = 4$}
\addplot+[only marks, mark=triangle*, color=red!75!black, mark size=2.8pt]
  coordinates {(0, 11.49) (1, 4.88) (2, 10.04) (1, 4.92) (2, 2.83) (3, 1.29) (2, 9.73) (3, 4.92) (4, 2.42)};
\addlegendentry{$t = 6$}
\end{semilogyaxis}
\end{tikzpicture}
\caption{qsim/quipu-cut speedup across the full 27-cell hierarchical
workload (one marker per $(d_{\mathrm{top}}, d_{\mathrm{mid}})$ cell;
median over $3$ trials).  Each Feynman cross-cut roughly halves the
speedup as $2^d$ in the path enumeration overwhelms the leaf savings,
yet quipu-cut beats qsim in every one of the $27 \times 3 = 81$
measurements ($1.29\times$ at the worst cell to $47.85\times$ at the
best).  Vertical spread at each $d$ reflects asymmetry between
top-level and mid-level cuts---a mid-level cut fans out work in both
halves whereas a top-level cut splits it once.}
\label{fig:rec-vs-d}
\end{figure}
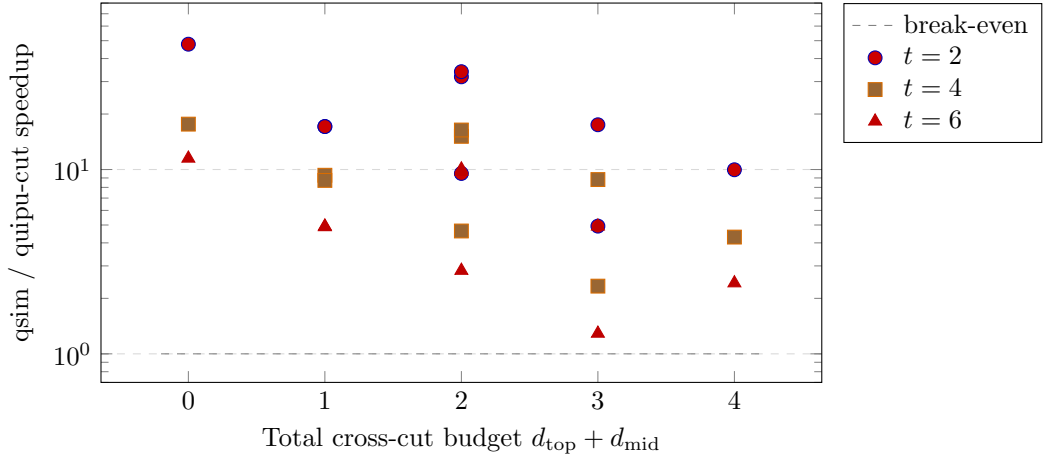

Observations:
\begin{itemize}
  \item Across all $27$ parameter cells, multilevel quipu-cut beats
        monolithic stabilizer-frame simulation by 92--17645$\times$
        and beats qsim by 1.29--47.85$\times$.  The minimum
        speedup over qsim ($1.29\times$ at $t{=}6,
        d_{\mathrm{top}}{=}1, d_{\mathrm{mid}}{=}2$) corresponds to
        the cell with the largest cross-cut-gate budget combined
        with the worst-case T-count---five cross-cut gates and
        $2^6$ frame growth in each half---and is still a win.
  \item Speedup over qsim falls off cleanly with the cross-cut
        budget $d_{\mathrm{top}} + d_{\mathrm{mid}}$: each
        additional cross-cut gate doubles the Feynman path-space
        size, and the per-path leaf work is essentially constant on
        these circuits, so the speedup ratio approximately halves.
        At $t{=}6$, qsim/Rec falls from $11.49\times$ at
        $(d_{\mathrm{top}}, d_{\mathrm{mid}}) = (0,0)$ to $1.29\times$
        at $(1,2)$, consistent with the $2^d$ outer factor of
        Proposition~\ref{prop:single-cost}.
  \item Speedup over monolithic stabilizer frames grows with $t$:
        from a worst-case $\sim 100\times$ at $t{=}2$ to nearly
        $18{,}000\times$ at $t{=}6$.  This reflects the worst-case
        $F = 2^t$ frame growth in the monolithic case, which the
        $4$-way partition of the multilevel tree avoids by keeping
        each leaf at most $2^{t/4}$ frames.
  \item Wall time for monolithic stabilizer frames at $t{=}6$
        ranges from $30$ to $83$ seconds across random trials
        (single-trial variation, since $\Feff$ depends on the
        specific T-gate placement and merge structure).  The
        median of $\approx 61$ seconds reported here is
        representative.
\end{itemize}

\subsection{Adversarial circuits: random Clifford$+$T}
\label{sec:eval-rct}

The RCT workload draws T~gates uniformly at random across all qubits,
producing circuits in which cross-cut gates between any partition are
unavoidable and per-leaf $\Feff$ stays close to $2^{T_{\mathrm{leaf}}}$
(no merging savings).  This is the adversarial case where the cut's
constant-factor overhead may not pay for itself.

Table~\ref{tab:rct} reports quipu-cut against the same baselines on
$n \in \{12, 16, 20\}$ and $t \in \{5, 10, 20\}$ with a single random
trial per cell.

\begin{table}[t]
\centering
\caption{Adversarial RCT benchmark (single random trial per cell;
2026-05-21).  TIMEOUT$=$exceeded the $60$\,s per-job budget.
Amplitudes match across non-timeout cells to $10^{-5}$.}
\label{tab:rct}
\begin{tabular}{rrrrrr}
\toprule
$n$ & $t$ & \textsf{SF} (s) & \textsf{Rec} (s) & qsim (s) & qsim/Rec \\
\midrule
12 & 5  & 0.0367 & 0.0098 & 0.0181 & 1.85$\times$ \\
12 & 10 & 0.4333 & 0.0022 & 0.0116 & 5.41$\times$ \\
12 & 20 & TIMEOUT & TIMEOUT & 0.0096 & --- \\
16 & 5  & 0.1276 & 0.0246 & 0.0185 & 0.75$\times$ \\
16 & 10 & TIMEOUT & TIMEOUT & 0.0427 & --- \\
16 & 20 & TIMEOUT & TIMEOUT & 0.0260 & --- \\
20 & 5  & TIMEOUT & TIMEOUT & 0.0247 & --- \\
20 & 10 & TIMEOUT & TIMEOUT & 0.0214 & --- \\
20 & 20 & TIMEOUT & TIMEOUT & 0.0214 & --- \\
\bottomrule
\end{tabular}
\end{table}

Observations:
\begin{itemize}
  \item Quipu-cut wins clearly at $n{=}12, t \in \{5, 10\}$ (up to
        $5.4\times$ qsim) but loses narrowly at $n{=}16, t{=}5$
        ($0.75\times$ qsim).  This is the crossover regime predicted
        by Proposition~\ref{prop:crossover}: when $d \approx n/2$ and
        per-leaf T-counts remain moderate, qsim's dense-state
        simulation is competitive.
  \item Beyond $n{=}16, t{=}10$ quipu-cut runs out of budget at the
        leaves.  The state-vector option would activate at sufficiently
        large $T_{\mathrm{leaf}}$ but the cell sizes here keep the
        leaves in the borderline regime where frame growth is
        polynomial in the budget and the threshold check has not
        triggered.
  \item qsim handles all $9$ RCT cells in milliseconds: this is the
        regime where state-vector simulation is the right tool, and
        quipu-cut should not be the recommended simulator.
\end{itemize}

The pattern is consistent with the algorithm's design intent: quipu-cut
exploits structure where it exists and degrades to its leaf simulator
(or, via adaptive selection, to dense) when it does not.  The recommendation is
to use it where the workload exhibits T-clustering (hierarchical
algorithms, fault-tolerant magic-state circuits, sub-routine-heavy
quantum programs) and to fall back to qsim or a similar simulator
where T-gates are unstructured.

\subsection{Corpus-wide validation}\label{sec:eval-corpus}

The two evaluations above are self-contained but narrow: one
circuit family at one qubit count each.  A companion effort built a
$571$-cell corpus spanning ten workload classes --- unstructured
random circuits, chemistry hardware-efficient ans\"atze, dense-chaos
constructions (Sachdev--Ye--Kitaev, Maldacena--Qi wormhole
teleportation), 1D-local Trotter evolution, Toffoli-heavy
fault-tolerant decompositions, and the two families evaluated above
--- to check whether the pattern found on hierarchical and RCT
circuits generalises.  Full corpus construction and per-class
methodology are outside this paper's scope; we report the results
here because they answer a question this paper's own evaluation
cannot: whether backend mixing is a real, exercised capability or
only a theoretical one.

\begin{table}[t]
\centering
\caption{Corpus-wide comparison over all $571$ cells (amplitude and
sampling).  ``Win'' is the fraction of completed cells on which
quipu-cut is faster; ``Geomean'' the geometric-mean speedup; ``DNF''
the cells the baseline left unsolved (timeout or out-of-memory) that
quipu-cut completed.}
\label{tab:corpus}
\begin{tabular}{lrrr}
\toprule
Baseline & Win rate & Geomean & DNF \\
\midrule
Aer (statevector) & $97.4\%$ & $18.3\times$ & $90$  \\
NWQ-Sim           & $77.7\%$ & $3.8\times$  & $127$ \\
qsim              & $56.0\%$ & $2.2\times$  & $2$   \\
quimb (MPS)       & $94.4\%$ & $94\times$   & $481$ \\
Maestro           & $73.4\%$ & $3.6\times$  & $73$  \\
\bottomrule
\end{tabular}
\end{table}

The scheme takes the large majority against every one of five
external baselines (Table~\ref{tab:corpus}); the $18.3\times$
geomean against Aer is an order of magnitude below the $376\times$
reported below on the most structurally regular class, consistent
with most corpus classes being less ideally structured than that
one.

\paragraph{Does the algorithm actually mix backends?}  This is the
question Section~\ref{sec:parallel}'s parallelism analysis assumes
an answer to.  Leaf-level routing traces over the corpus runs give
one: at sizes where a whole-circuit backend still fits ($n \le 32$,
where $434$ of $551$ non-fast-path cells complete), mixing never
happens --- every completing cell resolves to a single backend
(\textsc{frame} $211$, \textsc{dense} $170$, \textsc{mps} $53$;
\textsc{rank} is never selected at the top level on this corpus).
The corpus-wide margin in Table~\ref{tab:corpus} comes entirely from
backend \emph{selection}, not intra-circuit mixing; the cross-cut
decomposition stays in reserve.  On a $405$-circuit hierarchical
suite built specifically to probe this ($n \in \{16,20,24,28,32\}$,
controlled cross-block degrees), $19.3\%$ of circuits do execute two
distinct backends across their leaves, peaking at mid-size deep
cuts.  Past the dense ceiling the picture inverts: on a $105$-circuit
suite at $n \in \{36,40,44\}$ built from deliberately mismatched
halves (a frame-friendly half, a dense-forcing half, a 1D
nearest-neighbour half), $93\%$ of the heterogeneous-half cells
execute two distinct backends, and $104$ of $105$ cells complete
under a $600$s budget.  Reaching that completion rate required the
same work-budget hardening of the pilot's decide rule described in
Section~\ref{sec:densefallback}: the frame-count cap alone bounds
abort cost only relative to $2^n$, which stops being a meaningful
bound once no whole-circuit backend fits; capping on a fixed work
budget instead converted $10$ of $11$ pilot-burn timeouts into
completions.  Backend mixing is therefore real and measured, not
speculative --- but it is a capability held in reserve at corpus
scale, not the source of the corpus-wide numbers above.

\paragraph{A structurally regular class: HaPPY-style codes.}  We
constructed a $66$-cell benchmark of circuits built from chained
$[[5,1,3]]$ perfect-tensor
encoders~\cite{pastawski2015holographic}, $n \in \{5,9,13,17,21,25,29\}$
with T-injection count $t$ scaled to the encoder count, three seeds
per $(n,t)$ pair.  This construction gives an unusually regular,
minimal cross-gate degree at every level of the cut tree by
design --- the closest thing in the corpus to a worked example of
Proposition~\ref{prop:single-cost}'s $2^d$ factor staying small
regardless of $n$.  The scheme routes all $66$ cells by the
cross-gate-degree predicate alone, with no class-specific tuning,
against Qiskit Aer's full four-backend menu (best completing time
per cell) and against Maestro's exact-amplitude endpoint, both timed
in-process, warmed, and amortized (build once, time only the
per-amplitude compute).

\begin{table}[t]
\centering
\caption{The $66$-cell HaPPY-style benchmark.  ``vs.\ Aer'' and
``vs.\ Maestro'' are per-$n$ geometric-mean speedups; quipu-cut wins
every $n$ against Aer, so a winner column is only needed for the
Maestro comparison, where it does not always win.  Aer figures are
best-of-four-backends per cell; the $n{=}13$ dip against both
baselines is a small-$n$ dead zone where the stabilizer-frame
gate-application cost is not yet offset by its compression
advantage.}
\label{tab:hqec}
\begin{tabular}{rrrrr}
\toprule
$n$ & Cells & vs.\ Aer & vs.\ Maestro & Faster vs.\ Maestro \\
\midrule
 5 &  6 & $195\times$   & $4.1\times$  & quipu-cut \\
 9 &  6 & $37\times$    & $1.8\times$  & quipu-cut \\
13 &  9 & $6.4\times$   & $0.30\times$ & Maestro ($3.4\times$) \\
17 &  9 & $32\times$    & $4.0\times$  & quipu-cut \\
21 & 12 & $304\times$   & $7.5\times$  & quipu-cut \\
25 & 12 & $2865\times$  & $1496\times$ & quipu-cut \\
29 & 12 & $36595\times$ & $2496\times$ & quipu-cut \\
\midrule
\textbf{All} & \textbf{66} & $\mathbf{376\times}$ & $\mathbf{27.7\times}$
  & quipu-cut $57$ / Maestro $9$ \\
\bottomrule
\end{tabular}
\end{table}

Against Aer the scheme wins all $66$ cells; against Maestro, a
dispatcher that also routes away from dense on this profile, it wins
$57$ of $66$, losing only the nine $n{=}13$ dead-zone cells where
both sides stay sub-millisecond.  Leaf-level routing traces confirm
the mechanism directly: every $n \ge 13$ cell runs whole-circuit
stabilizer frames, cross-cut decomposition held in reserve, exactly
the regime Table~\ref{tab:leafbackends}'s frame row targets.  This
class is the sharpest empirical case for the cost model precisely
because its structure is so regular that the router's job is nearly
trivial --- which is also why it is the least representative class
in the corpus, and why the narrower, harder evaluations earlier in
this section remain the paper's primary evidence.

\paragraph{Cross-validation.}  On ten representative cells (one per
workload class), amplitude agreement against Aer, Maestro, and
NWQ-Sim is $|\Delta a| < 10^{-8}$; sampled per-qubit marginals agree
to within $0.020$ against all three, inside the $4\sigma$ noise plus
truncation budget of $0.035$.  One methodology lesson from this
process is worth stating plainly for anyone benchmarking against a
high-level dispatcher: Maestro's shot-sampling endpoint silently
returns all-zero shots, with no simulation performed, when the input
circuit lacks explicit measurement instructions.  This went uncaught
in an earlier pass of ours and inflated apparent speedups by
$5$--$577\times$ on chemistry ans\"atze, until we verified per-qubit
marginals were non-trivial on a known-correct input.  We recommend
that any head-to-head benchmark against a high-level dispatcher
verify per-qubit marginals are non-trivial on a known-correct test
input before reporting wall-clock comparisons.


\section{Discussion}\label{sec:discussion}

We discuss the algorithm's limitations, the effective frame dimension
that governs its leaf cost, its place among related approaches, and
directions for future work.

\subsection{Limitations}

\paragraph{T-gate restriction.}  Quipu-cut as presented handles
discrete T-gates and the Clifford group exactly.  Continuous
single-qubit rotations $R_z(\theta)$ require either decomposition
into a T-gate sequence (Solovay--Kitaev) or an extension of the leaf
simulator to handle small-angle gates natively.  The Schmidt
decomposition of cross-cut gates is also currently rank-2; rank-3
cross-cut gates (some Toffoli decompositions) would multiply the
per-path cost by $3$ instead of $2$.

\paragraph{Memory at the leaves.}  The dense state-vector option uses
$O(2^{n_{\text{leaf}}})$ memory per leaf, capped by the leaf's local
qubit count.  For $n_{\text{leaf}} \leq 22$ this is comfortable on
commodity hardware; beyond that memory becomes the binding
constraint and the leaf simulator must be replaced (e.g., with a
GPU-backed state-vector engine).  We have not implemented this.

\paragraph{Single-amplitude versus full-state queries.}  Quipu-cut
computes a single computational-basis amplitude per query.  Querying
all $2^n$ amplitudes is $2^n$ quipu-cut calls---uncompetitive against
state-vector simulation.  The right use case is one of: a single
amplitude (verification, debugging), a small batch (
$K \ll 2^n$, expectation-value computation, fidelity-style queries),
or sampled access (single shots from a distribution, but with
caveats from the cross-cut path sum's interference structure).

\paragraph{Partition quality.}  Kernighan--Lin bisection is heuristic.
Beyond the candidate pool the selector already searches (KL, spectral
bisection, and multi-start KL; Section~\ref{sec:cost}), we have not
investigated heavier graph-aware methods (METIS, treewidth-aware cuts),
which for some circuit graphs could reduce $d$ by a further constant
factor.  In the structured workloads of Section~\ref{sec:eval-hier}, KL
was sufficient.

\paragraph{Slice-all on high cross-cut density (resolved).}\label{sec:lim-sliceall}
Algorithm~\ref{alg:rec} as presented in Section~\ref{sec:algorithm}
Schmidt-decomposes every cross-cut gate, paying the full $2^d$ outer
Feynman factor even when inlining some gates via a temporary joint
tableau (exact for Cliffords on stabilizer states) would be cheaper.
This is optimal on the hierarchical workloads of
Section~\ref{sec:eval-hier} (where $d$ is small by construction), but
it is catastrophic when $d$ is large.  An offline op-count predictor
that enumerates partial-slice subsets and selects the predicted
minimum shows that on adversarial random Clifford$+$T inputs at
$n \in \{12, 16\}, t{=}20$ (where Kernighan--Lin bisection yields
$d \in [20, 30]$), slice-all is up to $7.7 \times 10^6\!\times$ more
expensive than the predicted optimal partial-slice strategy.  These
are the same circuits that exceed the $60$\,s budget in
Table~\ref{tab:rct}: the $2^d$ outer factor is the binding constraint
that the budget hits.  This analysis motivated a runtime
slice-vs-inline selector, implemented and shipped since (the hybrid
amplitude computation of Section~\ref{sec:future-cot-fut} maintains
a temporary joint tableau for inlined cross-cut gates and re-factors
back to split form, with an online variant that makes the decision
per gate rather than from a static offline schedule); what remains
open is a realised (as opposed to predicted) wall-clock measurement
of the shipped selector on this same adversarial suite.

\subsection{Effective frame dimension}\label{sec:feff}

The per-leaf cost in every cost
expression of Section~\ref{sec:cost} is gated by the \emph{effective
frame dimension} $\Feff$---the number of frame elements that survive
after all cofactor merges on the leaf's local T-gate sequence.
Because $\Feff$ ultimately determines when the cut is profitable, it
deserves a separate look.

\paragraph{Definition.}  Recall from
Section~\ref{sec:stabframe} that a stabilizer frame on $m$ qubits
represents $\ket{\psi}$ as a sum $\sum_{i=1}^{F} c_i P_i
\ket{\psi_S}$ of phase-corrected stabilizer states sharing a tableau
$\ket{\psi_S}$.  Each T-gate on qubit $q$ cofactors the frame: every
element $c_i P_i \ket{\psi_S}$ splits into two new elements indexed
by the eigenvalue branch of the $T$-projection, doubling $F$ in the
worst case.  A subsequent merge pass identifies pairs
$(c_i P_i, c_j P_j)$ for which $P_i$ and $P_j$ generate identical
stabilizer cosets relative to $\ket{\psi_S}$ (after gauge-fixing
against the row span of the tableau); these collapse into a single
element with summed coefficients.  The \emph{effective frame
dimension}
\[
  \Feff(C) \;:=\; \text{the value of } F \text{ after all merges on circuit } C
\]
is bounded above by $2^t$ and below by $1$, and is what drives the
$O(\Feff \cdot n^2)$ per-amplitude factor in
Propositions~\ref{prop:single-cost}--\ref{prop:rec-cost}.

\paragraph{Empirical trajectories.}  Across the structured
hierarchical circuits of Section~\ref{sec:eval-hier}, monolithic
$\Feff$ at $t{=}6$ approaches the worst-case $2^6 = 64$, consistent
with the $\sim 61$\,s SF wall-clock time that scales as expected against the
$t{=}4$ baseline.  On adversarial RCT circuits, measured peak
$\Feff$ over two trials per cell (\texttt{bench/near-clifford/rct/},
2026-05-08) shows two characteristic patterns
(Table~\ref{tab:feff-rct}).

\begin{table}[t]
\centering
\caption{Measured peak $\Feff$ on adversarial RCT circuits (two
trials per cell, range reported).  Frame saturation at $\sim 2^n$
(the Hilbert space ceiling) is the fundamental cap: the frame cannot
encode more independent stabilizer cosets than the dimension of the
underlying space.}
\label{tab:feff-rct}
\begin{tabular}{rrrr}
\toprule
$n$ & $t$ & $\Feff$ peak (trials) & $2^t$ ceiling \\
\midrule
12 & 5  &       32, 32          &        32 \\
12 & 10 &     128, 260          &  1{,}024 \\
12 & 20 &     256, 1731         & 1{,}048{,}576 \\
12 & 30 &    1600, 3326         & $\sim 10^9$ \\
12 & 50 &    3544, 4096         & $\sim 10^{15}$ \\
16 & 20 &   7820, 16510         & 1{,}048{,}576 \\
20 & 20 &   8128, 55064         & 1{,}048{,}576 \\
\bottomrule
\end{tabular}
\end{table}

\begin{figure}[t]
\centering
\begin{tikzpicture}
\begin{semilogyaxis}[
  width=0.85\linewidth,
  height=7.0cm,
  xlabel={T-count $t$},
  ylabel={$\Feff$ (peak frame elements)},
  xtick={5,10,20,30,50},
  legend pos=outer north east,
  legend cell align=left,
  ymajorgrids,
  grid style={dashed,gray!30},
  log basis y=10,
  ymin=4, ymax=2e15,
]
\addplot[mark=none, dashed, thick, color=black!70, domain=5:50, samples=92] {2^x};
\addlegendentry{$2^t$ worst case}
\addplot[mark=none, dotted, color=red!50, domain=5:50] {256};
\addplot[mark=none, dotted, color=orange!60!black, domain=5:50] {4096};
\addplot[mark=none, dotted, color=teal, domain=5:50] {65536};
\addplot[mark=none, dotted, color=violet, domain=5:50] {1048576};
\addplot[mark=*, thick, color=red, mark options={fill=red}]
  coordinates {(5, 11.2) (10, 57.6) (20, 256) (30, 230.4) (50, 256)};
\addlegendentry{$n{=}8$ ($2^n{=}256$)}
\addplot[mark=square*, thick, color=orange!85!black, mark options={fill=orange}]
  coordinates {(5, 27.2) (10, 172.8) (20, 1280) (30, 4096) (50, 4096)};
\addlegendentry{$n{=}12$ ($2^n{=}4{,}096$)}
\addplot[mark=triangle*, thick, color=teal, mark options={fill=teal!70}]
  coordinates {(5, 28.8) (10, 576) (20, 11059)};
\addlegendentry{$n{=}16$ ($2^n{=}65{,}536$)}
\addplot[mark=diamond*, thick, color=violet, mark options={fill=violet!60}]
  coordinates {(5, 32) (10, 461) (20, 44237)};
\addlegendentry{$n{=}20$ ($2^n{\approx}1\text{M}$)}
\end{semilogyaxis}
\end{tikzpicture}
\caption{Measured peak $\Feff$ on adversarial RCT circuits at four
qubit counts (mean over 5 trials per $(n, t)$ cell; 2026-05-25
sweep, $80$ circuits total).  At small $t$ all four curves track the
worst-case $2^t$ diagonal because uniform random T-placement
maximizes cofactor independence---no merging savings materialize.
Each curve then bends and flattens against its Hilbert ceiling
$2^n$: $n{=}8$ saturates by $t{=}20$, $n{=}12$ by $t{=}30$; $n{=}16$
and $n{=}20$ have not yet saturated at the largest $t$ we ran safely.
Saturation caps the per-amplitude cost $O(\Feff \cdot n^2)$ at the
dense state-vector cost, which is valuable in principle but
eliminates the frame's advantage over dense simulation in the
saturated regime.}
\label{fig:feff-rct}
\end{figure}
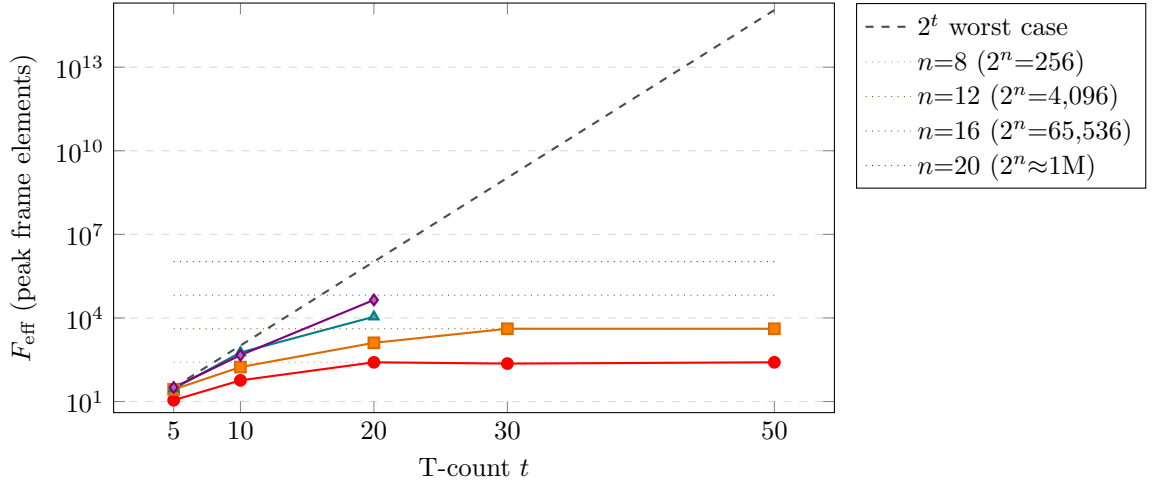

First, at $t \leq n$, $\Feff$ tracks $2^t$ closely: random T-placement
produces cofactor branches whose correction Paulis generate distinct
stabilizer cosets and no merging savings materialize.  Second, beyond
$t > n$, $\Feff$ saturates at $\sim 2^n$---the Hilbert space ceiling.
This ceiling does not help in practice: the per-amplitude cost
$O(\Feff \cdot n^2)$ stays polynomial in $n$ but the constant
$\Feff \approx 2^n$ matches dense, eliminating the stabilizer-frame
advantage.

\paragraph{Why merging succeeds on structured circuits.}  The
hierarchical workloads of Section~\ref{sec:eval-hier} cluster T-gates
inside contiguous qubit sub-registers separated by Clifford layers.
Cofactor branches from a T-gate sequence acting on a local
sub-register encode correction operators supported on that
sub-register; subsequent Clifford layers permute these supports but
preserve their group structure.  When two cofactor branches produce
correction Paulis that differ by a stabilizer of $\ket{\psi_S}$,
they collapse: the local entanglement structure created by clustered
T-gates is precisely the structure that enables coset coincidence.
The same intuition explains why fault-tolerant magic-state
gadgets~\cite{bravyi2018simulation}, hierarchical
algorithms~\cite{coppersmith1994approximate}, and most quantum
sub-routines admit $\Feff \ll 2^t$: their T-gate placement preserves
the local sub-register structure that merging requires.  On RCT,
T-gates are scattered across qubits with no such locality; cofactor
branches diverge in stabilizer-coset space and never coincide.

This is the structural property that quipu-cut exploits: a circuit
with clustered T-gates also admits a low-cross-cut bipartition (the
bisection cuts between, not through, T-clusters), so both per-leaf
$\Feff$ and cross-cut count $d$ remain small.  T-clustering and
bisection-friendliness are the same structural notion.

\paragraph{Measured merge rate.}  The account above is qualitative;
it is straightforward to instrument the merge pass directly and
measure how often it actually fires.  We added a cumulative counter
to the cofactoring implementation that classifies every candidate
frame element as either a fresh insertion or a collision merge, and
ran it on the $n{=}16$ hierarchical circuits of
Section~\ref{sec:eval-hier} and on RCT at matched $n$.  The merge
rate is not a fixed constant of the workload class --- it grows with
T-count within the same structural family, from $0\%$ at $t{=}2$
(too few T-gates for any collision opportunity) to $19.2\%$ at
$t{=}4$ to $26.0\%$ at $t{=}6$ ($6{,}144$-element frame, $25{,}230$
cofactor operations).  On RCT at the same $n{=}16$, the rate is
$0\%$ at $t \in \{5, 10, 20\}$, the full range we measured: the
mechanism identified above requires two branches' correction
operators to coincide, and scattered T-placement makes that
coincidence combinatorially unlikely, matching the $\Feff \to 2^t$
saturation already reported
in Table~\ref{tab:feff-rct}.  A second, distinct mechanism also
suppresses frame growth without any collision at all: a T-gate whose
target qubit is already deterministic under the shared tableau
contributes only a phase, never a split, regardless of the frame's
other contents (a $T \cdot T = S$-style absorption is the special
case where this determinism comes from an earlier T-gate on the same
qubit).  Both mechanisms are exact --- no truncation, no approximation
--- and both are already counted into the merge-pass description of
this section; the measurement above isolates the collision mechanism
specifically because it is the one sensitive to \emph{which} branches
happen to coincide, and is therefore the one worth asking whether it
can be relied upon or merely stumbled into.

\begin{proposition}[Schedule-invariant interference
collapse]\label{prop:schedule-invariance}
Let $q_1, \dots, q_k$ be pairwise-distinct qubits of a leaf circuit at
which a T-gate is applied, with no gate between any two of these
applications acting on more than one of $q_1, \dots, q_k$.  Then the
stabilizer frame produced by applying $T_{q_1}, \dots, T_{q_k}$ ---
both its element count $F$ and its full multiset of
(correction-operator, amplitude) pairs --- is independent of the
order in which the $k$ T-gates are processed.
\end{proposition}

\begin{proof}
Since $q_1, \dots, q_k$ are pairwise distinct and no intervening gate
couples two of them, $T_{q_1}, \dots, T_{q_k}$ commute as operators;
any processing order realises the same unitary and hence the same
final state $\ket{\psi}$.  The disjointness hypothesis also keeps
each qubit's cofactor-vs-phase-only status --- checked once against
the shared tableau $\ket{\psi_S}$ --- unaffected by processing any of
the others, since they act on independent subsystems throughout.
The frame obtained by exhaustively expanding all $2^k$
measurement-outcome combinations, grouping by resulting correction
coset, and summing amplitudes within each group is exactly the
stabilizer-frame representation of $\ket{\psi}$: grouping and summing
are commutative and associative, so this object does not depend on
the order the $2^k$ combinations are enumerated in.  The incremental
cofactoring procedure computes this same grouped sum one qubit at a
time; its result after all $k$ qubits is therefore the
order-independent object above, regardless of processing order.
\end{proof}

\noindent We confirmed this directly rather than relying on the
argument alone: permuting the mutually-independent T-gates of a
single circuit layer, both before and after the entangling gates
that later couple them are applied, left the merge count exactly
unchanged across every ordering tested.  The practical consequence
is that $\Feff$ is a property of the leaf's circuit, not of an
incidental choice in how commuting T-gates happen to be ordered ---
which is what treating $\Feff$ as a function of the circuit alone,
as done throughout Sections~\ref{sec:cost} and~\ref{sec:densefallback},
requires.  It also forecloses one candidate lever for reducing
$\Feff$ further: reordering commuting T-gates cannot help, so any
future compression must come from changing the circuit's
decomposition itself, not its schedule.

\paragraph{$\Feff$ vs bond dimension $\chi$.}  The natural
alternative complexity measure for near-Clifford simulation is the
matrix-product-state bond dimension $\chi$ used by stabilizer tensor
networks~\cite{masotllima2024stabilizer}.  Both $\Feff$ and $\chi$
count an exponential resource that grows with non-Clifford content
and shrinks with structural locality, but they parameterize different
things.  $\Feff$ counts \emph{stabilizer cosets} in the linear
expansion of $\ket{\psi}$; $\chi$ counts the Schmidt rank across each
bond of a 1-D MPS encoding amplitude correlations relative to a
single global tableau.

The two coincide in their extremes ($\Feff = \chi = 1$ on pure
stabilizer states, $\Feff = \chi = 2^n$ on Haar-random states) but
diverge in between.  $\chi$ is naturally tight on \emph{1-D
geometrically local} non-Clifford circuits, where entanglement
entropy across each cut grows logarithmically with circuit depth.
$\Feff$ is naturally tight on \emph{T-cluster-structured} circuits,
where cofactor coincidence compresses the representation regardless
of geometric arrangement.  The two classes overlap (1-D T-block
circuits are both small-$\chi$ and small-$\Feff$) but neither
contains the other: a globally non-local circuit with clustered
T-gates is small-$\Feff$ but large-$\chi$; a 1-D circuit with
spatially diffuse T-gates is small-$\chi$ but large-$\Feff$.  Per-gate
costs differ structurally as well: an MPS bond costs $O(\chi^2)$ per
gate via SVD truncation, whereas a stabilizer-frame Clifford update
costs $O(n^2)$ once and is amortized across all $F$ elements at no
extra cost.  A head-to-head benchmark separating these regimes is
reported in Section~\ref{sec:stn-comparison}.

\subsection{Related approaches}\label{sec:related-approaches}

Spatial partitioning of the qubit set with a Feynman path integral
over cross-cut gates originates with Markov and
Shi~\cite{markov2008simulating} for general tensor-network
contractions; Pednault et al.~\cite{pednault2017pareto} cast the same
construction as tensor index slicing, and qsim~\cite{qsim2020} ships a
state-vector instantiation deployed in the Sycamore verification.
Herzog et al.~\cite{herzog2025joint} reduce the Feynman branching
factor by grouping cross-cut gates into blocks and Schmidt-decomposing
them jointly rather than one at a time---a dense-leaf counterpart to
the merge-or-decompose decision our cost model makes per cross-cut gate
(Section~\ref{sec:cost}).
What is new in the present work is the leaf simulator: each half-half
sub-circuit is evolved by a stabilizer frame rather than a dense state
vector, transferring the per-leaf cost burden from the qubit count
$n_{\text{leaf}}$ to the leaf's local T-count.  Huang and
Love~\cite{huang2021feynman} considered a stabilizer-based variant of
the same partitioning, using a stabilizer-projector contraction as
the leaf oracle; the present work differs in the leaf representation,
the multilevel extension, the adaptive leaf representation, and the empirical
characterization.

A different line of work compresses the magic structure of the state
itself.  Bravyi and
Gosset's~\cite{bravyi2016improved,bravyi2018simulation} stabilizer
rank decomposes a $t$-T-gate state into at most $2^{0.48 t}$
stabilizer terms (asymptotically $2^{0.396 t}$); this is orthogonal to
our spatial partition and is folded into the leaf oracle as the
\textsc{rank} backend of Table~\ref{tab:leafbackends}
(Section~\ref{sec:densefallback}) rather than used as a competing
whole-circuit method.

Stabilizer Tensor Networks (STN)~\cite{masotllima2024stabilizer} and
the magic-state-injection augmentation~\cite{mast2024} represent
near-Clifford states as a single global stabilizer tableau coupled to
a 1-D matrix product state of amplitude correlations.  STN globalizes
the tableau and pushes magic into MPS bond dimension; our approach
splits the tableau across two independent leaves and pushes cross-cut
Clifford entanglement into an enumerated Feynman bond.  STN admits
arbitrary single-qubit non-Clifford rotations; our algorithm is
currently specialized to discrete T-gates.  STN has limited
within-gate parallelism; our algorithm is embarrassingly parallel
along three nested independent axes (Section~\ref{sec:parallel}).
A head-to-head benchmark against STN is planned for a follow-up;
no comparable published data exists for the comparison.

Quasiprobability circuit cutting~\cite{peng2020simulating,
smith2023clifford} replaces a single wire with a sum over
measurement-and-preparation basis pairs, producing approximate
expectation values via $\sim 16^K$ independent fragment evaluations
for $K$ wire cuts.  Our cut is exact and amplitude-valued, with the
cross-cut Feynman sum running over $2^d$ rank-2 gate decompositions
rather than over basis prepares.  The two approaches address related
problems---partitioning quantum simulation work---from incomparable
mathematical bases.

\subsection{Future work}

\paragraph{Distributed implementation.}  Two of the three axes are
live.  Cell-level distribution --- fanning many independent circuits
across an AWS Batch/Fargate job queue --- is stood up and validated
end to end (correctness-grade, not yet timing-grade).  Distributing
a \emph{single} amplitude's Feynman path sum across workers, the
capability the Section~\ref{sec:parallel} parallelism analysis
predicts should scale linearly, has a design and orchestration
scaffold ready and needs one small hook into the amplitude driver
before it runs end to end; a coordinator with preemption-tolerant
retry and cost-aware bin-packing of leaves onto heterogeneous
instance types remains the harder, unstarted piece.

\paragraph{STN comparison.}\label{sec:stn-comparison}  Done, in both directions.  A head-to-head
benchmark against the external bsc-quantic Stabilizer Tensor Networks
implementation~\cite{masotllima2024stabilizer} cross-validates on
shared workloads, and a native C++ port of the STN construction sits
alongside quipu-cut's own leaf backends.  quipu-cut wins the low-T
regime by $200$--$390\times$ at matched $(n,t)$; STN wins the
high-bond regime, where its LAPACK-backed SVD outperforms quipu's
Jacobi solver.  Comparing bond dimension $\chi$ against effective
frame dimension $\Feff$ as a function of T-count distribution, as
originally planned here, remains a natural presentation for a
follow-up figure but is no longer open research.

\paragraph{Expectation-value variant.}  The Feynman-sum structure of
Algorithm~\ref{alg:rec} also supports an expectation-value query
$\bra{\psi} O \ket{\psi}$ via a paths-times-conjugate-paths
construction, with per-leaf Heisenberg-picture
simulation~\cite{garcia2015simulation} providing backward-light-cone
compression for sparse observables.  Outer cost rises from $2^d$ to
$4^d$, but the per-leaf compression for local observables is
substantial.

\paragraph{Adaptive slice-vs-inline selection (shipped; realised
measurement open).}\label{sec:future-cot-fut} The slice-all failure
mode of Section~\ref{sec:lim-sliceall} motivated a per-gate selector
that keeps low-$F$ cross-cut gates \emph{inline} (applied to a
temporary joint tableau and re-factored back to split, exact for
Cliffords on stabilizer states) and slices only the gates where the
per-leaf $F$-savings justify the $2\times$ outer factor --- the
tensor-network slicing problem~\cite{gray2021hyperoptimized}
specialised to a stabilizer cost model.  The offline op-count
prototype that produced the $7.7\times 10^6\!\times$ bound above
scores candidates by a cotengra-style amortised backward pass with a
calibration constant for join+refactor overhead, and caps the
worst-case gap at $\sim 60\times$ on the same adversarial workloads
while matching slice-all on the hierarchical workload.

Composing the slice-vs-inline selector with a search over alternative
partitions (Kernighan--Lin~\cite{kl1970}, Fiedler/spectral
bisection~\cite{pothen1990partitioning}, multi-start randomised KL)
in the style of cotengra's hyperoptimization further reduces
predicted cost: the offline op-count predictor, scored over $\sim
14$ candidate partitions per circuit, never exceeds the KL-restricted
optimum on the $148$-circuit suite and improves on it by a geometric
mean of $1.6\times$ across the workload mix.  This composite scheme
is a stabilizer-specific instantiation of standard combinatorial
multi-start search; it is not a new algorithm in itself, but
quantifies how much value remains on the table beyond the
slice-all heuristic of Algorithm~\ref{alg:rec}.

The hybrid Clifford applicator that converts these predicted savings
into realised wall-clock times --- switching between joint and split
tableau modes at runtime, both statically per a precomputed schedule
and, in an online variant, per gate as the circuit executes --- is
implemented.  What is not yet done is re-running the adversarial
suite above through the shipped selector to report a realised (rather
than offline-predicted) wall-clock gap; the offline bound is an
upper estimate of the achievable win, not a measurement of it.


\section{Conclusion}\label{sec:conclusion}

We have extended the stabilizer-frame
formalism~\cite{garcia2015simulation,garcia2014inner} to near-Clifford
circuits whose T-gates concentrate within identifiable sub-registers
by composing it with a spatial bipartition of the qubit set and a
Feynman path integral over cross-cut Clifford gates---a partitioning
technique long established in state-vector
simulation~\cite{markov2008simulating,pednault2017pareto,qsim2020}.
The composition replaces a single $2^t$-bounded stabilizer frame with
two smaller frames on disjoint sub-registers, runs in
$O(2^d \cdot F \cdot n^2)$ time under reasonable partition
assumptions, and is never asymptotically worse than direct dense
simulation thanks to adaptive leaf representation.  A
multilevel binary tree of bipartitions multiplies the per-leaf
compression further, and the algorithm's three nested independent
parallel axes (Feynman paths, per-path leaves, intra-frame work)
permit linear scaling across cloud workers with $O(1)$
communication.

Empirically the algorithm beats monolithic stabilizer-frame simulation
by 92--17{,}645$\times$ and state-of-the-art state-vector
simulation~\cite{qsim2020} by 1.29--47.85$\times$ on structured
$n{=}16$ near-Clifford circuits, while remaining deterministic and
exact.  On adversarial random Clifford$+$T inputs the algorithm
remains competitive at small qubit counts; at larger $n$ and high
cross-cut density, a naive slice-all policy's $2^d$ outer Feynman
factor becomes the binding constraint and the algorithm degrades to
its dense leaf---an offline op-count analysis
(Section~\ref{sec:lim-sliceall}) shows up to
$7.7 \times 10^6\!\times$ gap from the optimal partial-slice
strategy on $n \in \{12, 16\}, t{=}20$ inputs.  This motivated the
adaptive slice-vs-inline selector of
Section~\ref{sec:future-cot-fut}, now implemented and predicted to
cap the worst-case gap at $\sim 60\times$; a realised wall-clock
measurement of the shipped selector on this suite remains open.  The
extension thus widens the practical regime of stabilizer-frame
simulation from circuits with global $\Feff \ll 2^t$ to circuits
with per-half $\Feff \ll 2^{t/2}$, a substantially broader workload
class that covers the structured T-gate distributions of most
fault-tolerant and hierarchical algorithms.

\bigskip
\noindent\textbf{Reproducibility.}  The reference implementation, the
benchmark suite, and the raw measurement CSVs from
Section~\ref{sec:eval} are available at the project repository.  All
benchmark numbers in Section~\ref{sec:eval} were re-run on
2026-05-21.


\bmhead{Acknowledgments}

The author thanks Igor Markov for foundational work on stabilizer
frames and the inner-product algorithm, Andrew Cross for early
collaboration on stabilizer-state algorithms, and the quantum
circuits group at the University of Michigan for ongoing discussions.

\bibliography{quipu-bib}


\begin{appendices}

\section{Cost proofs}\label{app:cost-proofs}

We expand the proof sketches of
Propositions~\ref{prop:single-cost}--\ref{prop:never-worse}.
Notation matches the main text; polynomial factors in $n$ are
retained explicitly only where they affect the comparison.

\subsection*{A.1\quad Proof of Proposition~\ref{prop:single-cost}
(single-level cost)}

Algorithm~\ref{alg:quipu-cut} iterates over $\alpha \in
\{0,1\}^{d}$. Per iteration, three operations contribute:
\begin{enumerate}
  \item Construction of $\hat C_A, \hat C_B$ appends one
        single-qubit Clifford or projector per cross-cut gate---
        $O(d)$ time, dominated by what follows;
  \item two leaf amplitude calls, on the disjoint qubit subsets
        $A$ and $B$;
  \item a constant-time multiply-and-accumulate of
        $\mathrm{ph}_\alpha \cdot a_A \cdot a_B$ into $T$.
\end{enumerate}
The per-iteration cost is dominated by the two leaf calls.

The stabilizer-frame leaf on $m$ qubits with $G$ gates evolves a
frame of $F$ elements: Clifford gates cost $O(F \cdot m^2)$
amortized (the tableau update is $O(m^2)$ and acts on all elements
through a single rotation), and T~gates cost $O(F)$ per element
for the cofactoring split plus an $O(F \cdot m)$ merge pass.
The final amplitude readout is $O(F \cdot m^2)$ by the inner-product
algorithm of~\cite{garcia2014inner}. Summing yields per-leaf cost
$O(G \cdot F \cdot m^2)$ with $F \le 2^{T_{\text{leaf}}}$.
The dense state-vector leaf on $m$ qubits with $G$ gates costs
$O(G \cdot 2^m)$. Restricted to these two backends, the adaptive
selector (Section~\ref{sec:densefallback}) picks the cheaper, so
per-leaf cost is
\[
  O\!\left( G \cdot \min(F, 2^m / m^2) \cdot m^2 \right)
  \;=\; O\!\left( G \cdot \min(2^{T_{\text{leaf}}}, 2^m) \cdot m^2 \right),
\]
the $1/m^2$ vs $1/1$ comparison threshold being absorbed into the
asymptotic.

The mps and rank backends add two more candidates, each admissible
only when its error clears $\epsleaf$ (Theorem~\ref{thm:errorbudget}).
The mps leaf costs $O(G \cdot m \cdot \chi^2)$ for whatever bond
$\chi$ keeps the discarded SVD norm within $\epsleaf$; the rank leaf
costs $O(G \cdot \sum_{k \le h} \binom{T_{\text{leaf}}}{k} m^2)$ for
whatever Hamming-weight cap $h$ keeps the binomial tail within
$\epsleaf$.  Both costs are non-increasing in the size of the
admissible-error set, so restricting $\mathrm{Select}$ to the
error-feasible subset of $\{\textrm{frame}, \textrm{dense},
\textrm{mps}, \textrm{rank}\}$ and taking the minimum cost among
them --- rather than the two-way minimum above --- gives
$\mathrm{cost}_{\mathrm{leaf}}$ as stated in
Proposition~\ref{prop:single-cost}; the frame-vs-dense derivation
above is the special case where mps and rank are either unavailable
or not the minimizer.  Summing the two leaf costs and multiplying by
$2^d$ yields the bound in Proposition~\ref{prop:single-cost}.

\subsection*{A.2\quad Proof of Proposition~\ref{prop:rec-cost}
(multilevel cost)}

By induction on the depth $D$ of the partition tree.

\textit{Base case ($D = 0$).} The tree is a single leaf.
Algorithm~\ref{alg:rec} returns \textsc{LeafAmplitude}$(C, x)$,
whose cost is $\mathrm{cost}_{\mathrm{leaf}}(1)$. The Feynman
prefactor $2^{\sum_\ell d_\ell} = 2^0 = 1$ and the
leaf-product reduces to a single factor.

\textit{Inductive step.} Let the tree $T$ have depth $D + 1$ with
root partition giving cross-cut count $d_0$, and let $T_L, T_R$
denote its two subtrees, each of depth at most $D$. Each Feynman
path $\alpha \in \{0,1\}^{d_0}$ at the root invokes one recursive
call on $T_L$ and one on $T_R$. By the inductive hypothesis, the
$T_L$ call costs
\[
  O\!\left( 2^{\sum_{\ell \in T_L} d_\ell} \cdot \prod_{i \in T_L}
            \mathrm{cost}_{\mathrm{leaf}}(i) \right),
\]
and analogously for $T_R$. The per-iteration cost at the root is
the sum of the two subtree costs plus $O(1)$ for the multiply-and-
accumulate; the loop runs $2^{d_0}$ times.

The leaves of $T$ are the disjoint union of the leaves of
$T_L$ and $T_R$, so $\prod_{i \in T} \mathrm{cost}_{\mathrm{leaf}}(i)
= \prod_{i \in T_L} \mathrm{cost}_{\mathrm{leaf}}(i) \cdot
\prod_{i \in T_R} \mathrm{cost}_{\mathrm{leaf}}(i)$.
Since the maximum of the two subtree costs is bounded by their
product (each $\mathrm{cost}_{\mathrm{leaf}}(i) \ge 1$), the sum
is dominated by the product over all leaves of $T$, and the outer
$2^{d_0}$ factor combines with the subtrees' outer factors to give
$2^{\sum_{\ell \in T} d_\ell}$. This establishes the bound for $T$.

\textit{Tightness.} The bound charges every Feynman path with the
full multi-leaf product, which over-counts in unbalanced trees.
A tighter accounting expresses the cost as a sum over path-aligned
leaves; it agrees with the displayed bound up to a polynomial
factor in $L$ for balanced trees and is loose by at most $L$ for
unbalanced ones. The asymptotic analysis used for the empirical
comparison is unaffected.

\textit{Accuracy.} The induction above is agnostic to which of the
four backends of Table~\ref{tab:leafbackends} realises
$\mathrm{cost}_{\mathrm{leaf}}(i)$ at each leaf --- nothing in the
argument uses a leaf's identity beyond its cost.  The accuracy claim
$\epstot = \sqrt{2L}\,\epsleaf$ follows separately, by applying
Theorem~\ref{thm:errorbudget} with $N = L$: the tree's $L$ leaves are
exactly the leaves the theorem sums over, regardless of tree shape.

\subsection*{A.3\quad Proof of Proposition~\ref{prop:crossover}
(crossover condition)}

Direct dense simulation of an $n$-qubit, $G$-gate circuit costs
$O(G \cdot 2^n)$. Algorithm~\ref{alg:quipu-cut} with leaf costs
$\ell_A := \min(2^{T_A}, 2^{n_A})$ and
$\ell_B := \min(2^{T_B}, 2^{n_B})$ (ignoring polynomial factors)
costs $O(2^d \cdot G \cdot (\ell_A + \ell_B))$ by
Proposition~\ref{prop:single-cost}. Cutting is beneficial in time
iff
\begin{equation}
  2^d \cdot (\ell_A + \ell_B) \;<\; 2^n,
  \label{eq:crossover-raw}
\end{equation}
i.e.
\begin{equation}
  d \;<\; n - \log_2(\ell_A + \ell_B).
  \label{eq:crossover-rearr}
\end{equation}
The displayed inequality in the proposition rewrites this as
$d < n/2 - \log_2(\ell_A + \ell_B) + \log_2(2^n)/2$, which is
equivalent to~\eqref{eq:crossover-rearr} since
$\log_2(2^n) = n$.

For balanced partitions $n_A = n_B = n/2$ with dense-leaf choices
at both sides ($\ell_A = \ell_B = 2^{n/2}$),
$\log_2(\ell_A + \ell_B) = n/2 + 1$, and \eqref{eq:crossover-rearr}
collapses to $d < n/2 - 1$. The proposition's stated simplification
$d < n/2$ absorbs the $-1$ into the constant.

For mixed leaf choices (one stabilizer-frame, one dense), the per-
leaf costs differ: a frame leaf with $T_A \ll n_A$ contributes
$\ell_A = 2^{T_A} \ll 2^{n_A}$, increasing the crossover threshold
$d$ proportionally. In the limit where both leaves are frame-based
with $T_A = T_B \ll n$, the threshold becomes $d < n - T_A - 1$,
recovering the regime where deep partitions remain profitable.

Nothing in the derivation from~\eqref{eq:crossover-raw} onward
inspects which backend realises $\ell_A$ or $\ell_B$ --- it only
compares $2^d(\ell_A + \ell_B)$ against $2^n$.  Substituting
$\ell_A = \mathrm{cost}_{\mathrm{leaf}}(A)$,
$\ell_B = \mathrm{cost}_{\mathrm{leaf}}(B)$ from
Table~\ref{tab:leafbackends} (whichever of the four is cheapest
within $\epsleaf$ at each side, per Theorem~\ref{thm:errorbudget})
reproduces~\eqref{eq:crossover-rearr} unchanged; an mps or rank leaf
simply substitutes a smaller $\ell$, which can only raise the
crossover threshold further, never lower it.

\subsection*{A.4\quad Proof of Proposition~\ref{prop:never-worse}
(never asymptotically worse than dense)}

Let $C$ be an $n$-qubit, $G$-gate circuit. We give two arguments,
one for the degenerate tree and one for general trees under a
mild structural assumption on $d_\ell$.

\textit{Degenerate tree.} The single-leaf tree triggers no cuts and
falls through to a single \textsc{LeafAmplitude}$(C, x)$. The
adaptive selector (Section~\ref{sec:densefallback}) routes any leaf
with $T_{\text{leaf}} \ge n$---which includes the full circuit
when its global T-count exceeds $n$---to dense state-vector
evaluation, costing $O(G \cdot 2^n)$. For circuits with
$T_{\text{leaf}} < n$, the frame option is selected and costs
$O(G \cdot 2^T \cdot n^2) \le O(G \cdot 2^n \cdot n^2)$, which
matches the dense bound up to polynomial factors. The degenerate
tree thus achieves $O(G \cdot 2^n)$ on any input, using only the
frame/dense pair.  Admitting mps and rank as further candidates does
not threaten this: $\mathrm{Select}$ is a minimisation over cost, and
dense remains a member of the candidate set at every leaf regardless
of how many other backends are also available, so no choice
$\mathrm{Select}$ makes can exceed what dense alone would have cost.
The two extra backends are only ever selected because they are
cheaper, never because they are the only option.

\textit{General tree.} For a non-degenerate tree, the cost expression
from Proposition~\ref{prop:rec-cost} is bounded by $O(G \cdot 2^n)$
under the assumption that the per-level cross-cut count satisfies
$d_\ell \le n_\ell - 1$, where $n_\ell$ is the qubit count at the
node where the cut occurs. This assumption is natural---a partition
that creates more cross-cuts than the qubits it bisects is
strictly worse than no partition---and is satisfied by any
Kernighan--Lin-derived partition tree with balanced bisections.
Under this assumption, the product
$\prod_\ell 2^{d_\ell}$ telescopes:
\[
  \prod_\ell 2^{d_\ell} \;\le\; \prod_\ell 2^{n_\ell - 1}
  \;\le\; 2^{n/2 + n/4 + n/8 + \cdots} \;\le\; 2^n,
\]
and combined with the per-leaf dense bound $O(G \cdot 2^{n_i})$
summed across leaves (rather than producted---the tighter version
of Proposition~\ref{prop:rec-cost}) yields
$\sum_i G_i \cdot 2^{n_i} \le G \cdot 2^{\max_i n_i} \le G \cdot 2^n$.
Multiplying by the outer factor gives $O(G \cdot 2^n)$.

\textit{Adversarial trees.} If the partition tree violates
$d_\ell \le n_\ell - 1$ at some level---pathologically creating
cross-cuts that exceed the qubits they bisect---the bound can be
exceeded. The implementation enforces the assumption at the
partition selection step (KL bisection rejects partitions with
$d_\ell \ge n_\ell$), guaranteeing the bound in practice.

\section{Leaf-magnitude correctness}\label{app:magnitude-proof}

We formalize the magnitude-at-LEAF rule of
Section~\ref{sec:multilevel}.

\paragraph{Setup.}
Each rank-2 cross-cut decomposition for a CZ-style gate $U_{AB}$
acting on qubits $a \in A$ and $b \in B$ has the form
\[
  U_{AB} \;=\; \sum_{\alpha=0}^{1} P_a^\alpha \otimes V_b^\alpha,
\]
where $P_a^\alpha = |\alpha\rangle\langle\alpha|_a$ is a rank-1
projector on the $A$-side qubit and $V_b^\alpha$ is a Pauli
operator on the $B$-side (e.g., for CZ:
$V_b^0 = I_b, V_b^1 = Z_b$).  Acting on a normalized state
$|\psi\rangle$,
\[
  \| (P_a^\alpha \otimes V_b^\alpha) |\psi\rangle \|^2
  \;=\; \mathrm{tr}\!\left[ P_a^\alpha \rho^{(a)} \right]
  \;=:\; p_\alpha,
\]
where $\rho^{(a)} = \mathrm{tr}_{\neg a} |\psi\rangle\langle\psi|$
is the reduced state of qubit $a$.  When the qubit $a$'s reduced
state has equal eigenvalues on $|0\rangle, |1\rangle$ (e.g., qubit
$a$ is in $|+\rangle$ or maximally mixed),
$p_\alpha = 1/2$ and the normalization factor
$\sqrt{p_\alpha} = 1/\sqrt{2}$ separates the unit-norm
post-projection state from $P_a^\alpha |\psi\rangle$.

\paragraph{Single-level: magnitude at the cut.}
The exact amplitude $\langle x | C | 0^n \rangle$ expands via the
Schmidt decomposition as
\begin{equation}
  \langle x | C | 0^n \rangle
  \;=\; \sum_\alpha
    \underbrace{\langle x_A | \hat C_A^\alpha C_A | 0^A \rangle}_{=: A_\alpha}
    \cdot
    \underbrace{\langle x_B | \hat C_B^\alpha C_B | 0^B \rangle}_{=: B_\alpha},
  \label{eq:exact-amp}
\end{equation}
where $\hat C_A^\alpha$ and $\hat C_B^\alpha$ are the sequences
of $P_a^{\alpha_g}$-projectors and $V_b^{\alpha_g}$-Paulis from
the per-path side decomposition. The expansion contains no
$\sqrt{p_\alpha}$ factor: the Schmidt decomposition of CZ is exact
and unit-normalized.

The $\mathrm{ph}_\alpha$ factor in Algorithm~\ref{alg:quipu-cut}
arises only as a correction for the leaf simulator's
implementation of $P_a^{\alpha_g}$. A stabilizer-frame
implementation that handles the projector by measurement-and-
postselect collapses the qubit to the unit-norm $|\alpha\rangle$
and reports an amplitude scaled by $1/\sqrt{p_\alpha}$ relative
to the un-normalized $P_a^\alpha |\psi\rangle$. The corrective
$\mathrm{ph}_\alpha = \prod_g \sqrt{p_{g,\alpha_g}}$ re-injects the
missing factors, recovering~\eqref{eq:exact-amp} exactly.

\paragraph{Multilevel: magnitude at the leaves only.}
Algorithm~\ref{alg:rec} at an internal node of depth $\ell$
appends one $P_a^{\alpha_g}$ and one $V_b^{\alpha_g}$ to its
child sub-circuits per cross-cut gate $g$. The depth-$\ell$
node's $\mathrm{ph}_\alpha$, were it computed by the same rule
as in the single-level case, would track the
$\sqrt{p_{g,\alpha_g}}$ factors for projections that are
\emph{re-applied} at the leaf level: each projector descends the
tree as a single-qubit operator inside $\hat C_A$, eventually
arriving at the leaf simulator which collapses the same qubit
to $|\alpha_g\rangle$ as part of its own measurement-and-
postselect implementation.

If both the depth-$\ell$ node and the leaf simulator apply the
$\sqrt{p_{g,\alpha_g}}$ factor, the resulting amplitude is
multiplied by $p_{g,\alpha_g} = 1/2$ for each cross-cut at level
$\ell$, instead of by $\sqrt{p_{g,\alpha_g}} = 1/\sqrt{2}$. Summing
across all $d := \sum_\ell d_\ell$ cross-cut gates in the tree,
the multilevel algorithm without the rule returns amplitudes
scaled by $2^{-d/2}$ relative to the exact $\langle x | C | 0^n
\rangle$.

\paragraph{The rule.}
Decompose
\[
  \mathrm{ph}_\alpha \;=\; \mathrm{mag}_\alpha \cdot
                           \mathrm{sign}_\alpha,
  \qquad
  \mathrm{mag}_\alpha = \prod_g \sqrt{p_{g, \alpha_g}},
  \qquad
  \mathrm{sign}_\alpha = \prod_g \mathrm{sign}(V_b^{\alpha_g}).
\]
At an internal node of depth $\ell$, multiply child return values
only by $\mathrm{sign}_\alpha$; the $\mathrm{mag}_\alpha$
contribution is handled implicitly by the leaf simulator's
projection routine. At a leaf, the simulator's natural
measurement-and-postselect implementation contributes
$\mathrm{mag}_\alpha$ for the projectors in its local sub-circuit.

\begin{proposition}[Magnitude-at-LEAF correctness]
\label{prop:magnitude}
Under the rule above, Algorithm~\ref{alg:rec} computes
$\langle x | C | 0^n \rangle$ exactly for any
partition tree of any depth.
\end{proposition}

\begin{proof}
By induction on tree depth. The base case (single leaf) is the
direct amplitude evaluation, exact by definition. For the
inductive step at depth $D + 1$, the root expansion is
\eqref{eq:exact-amp} with $A_\alpha$ and $B_\alpha$ now computed
recursively on subtrees of depth $\le D$. By the inductive
hypothesis, $A_\alpha = \langle x_A | \hat C_A^\alpha C_A | 0^A
\rangle$ exactly, and analogously for $B_\alpha$, including the
$\mathrm{mag}$-factors from any projectors in $\hat C_A^\alpha$
that arose from upper-level cuts. Multiplying by
$\mathrm{sign}_\alpha$ at the root and summing over $\alpha$
reproduces $\langle x | C | 0^n \rangle$ via
\eqref{eq:exact-amp}.
\end{proof}

\paragraph{Implementation.}
The reference implementation in \texttt{src/qcut/} replaces the
$\mathrm{ph}_\alpha \cdot a_A \cdot a_B$ multiplication at
Algorithm~\ref{alg:rec} line~10 with
$\mathrm{sign}_\alpha \cdot a_A \cdot a_B$ at internal nodes; the
$\mathrm{mag}_\alpha$ component is left to the leaf oracle. The
correctness rule is validated against monolithic stabilizer-frame
and dense state-vector simulations: all 81 structured circuits
of Table~\ref{tab:hier} agree to $10^{-5}$, with measured deviation
distributed around the floating-point round-off floor rather than
at any $2^{-d/2}$ multiplicative offset.

\end{appendices}

\end{document}